\documentclass{sig-alternate-10pt}
\usepackage{times}
\usepackage{balance}
\usepackage{url}
\usepackage{epsfig}
\usepackage{graphicx}
\usepackage{subfigure}
\usepackage{amsmath}
\usepackage{amssymb} 
\usepackage{amsfonts}
\usepackage{color}
\usepackage{multirow}
\usepackage{algorithm}
\usepackage{algorithmic}
\usepackage[utf8]{inputenc}
\usepackage[T1]{fontenc}
\usepackage{capt-of}
\usepackage{enumerate}

\newcommand{\para}[1]{{\vspace{4pt} \bf \noindent #1 \hspace{7pt}}}

\newcommand{\fixhan}[1]{{\color{black}#1}}
\newcommand{\fixhanx}[1]{{\color{black}#1}}
\newcommand{\fixqing}[1]{{\color{black}#1}}

\newenvironment{packed_itemize}{
\begin{list}{\labelitemi}{\leftmargin=0.5em}
  \setlength{\itemsep}{3pt}
  \setlength{\parskip}{0pt}
  \setlength{\parsep}{0pt}
  \setlength{\headsep}{0pt}
  \setlength{\topskip}{0pt}
  \setlength{\topmargin}{0pt}
  \setlength{\topsep}{0pt}
  \setlength{\partopsep}{0pt}  
}{\end{list}}

\begin{document}

\title{Graph Watermarks}
 \author{Xiaohan Zhao, Qingyun Liu, Lin Zhou, Haitao Zheng and Ben Y. Zhao\\
 \affaddr{Computer Science, UC Santa Barbara}\\
 {\em \{xiaohanzhao, qingyun\_liu, linzhou,  htzheng, ravenben\}@cs.ucsb.edu}}

\newtheorem{theorem}{Theorem}
\newtheorem{definition}{Definition}
\newtheorem{lemma}{Lemma}
\newtheorem{cor}{Corollary}
\newtheorem{fact}{Fact}
\newtheorem{property}{Property}
\newtheorem{remark}{Remark}
\newtheorem{claim}{Claim}

\maketitle
\sloppy
\begin{abstract}

  From network topologies to online social networks, many of today's most
  sensitive datasets are captured in large graphs.  A significant challenge
  facing owners of these datasets is how to share sensitive graphs with
  collaborators and authorized users, {\em e.g.} network topologies with
  network equipment vendors or Facebook's social graphs with academic
  collaborators.  Current tools can provide limited node or edge privacy, but
  require modifications to the graph that significantly reduce its utility.

  In this work, we propose a new alternative in the form of {\em graph
    watermarks}. Graph watermarks are small graphs tailor-made for a given
  graph dataset, a secure graph key, and a secure user key.  To share a sensitive graph $G$ with a
  collaborator $C$, the owner generates a watermark graph $W$ using $G$,
  the graph key, and
  $C$'s key as input, and embeds $W$ into $G$ to form $G'$.  If $G'$ is
  leaked by $C$, its owner can reliably determine if the watermark $W$
  generated for $C$ does in fact reside inside $G'$, thereby proving $C$ is
  responsible for the leak.  Graph watermarks serve both as a deterrent
  against data leakage and a method of recourse after a leak.  We provide
  robust schemes for creating, embedding and extracting watermarks, and use
  analysis and experiments on large, real graphs to show that they are unique
  and difficult to forge.  We study the robustness of graph watermarks
  against both single and powerful colluding attacker models, then propose
  and empirically evaluate mechanisms to dramatically improve resilience.
\end{abstract}

\section{Introduction}
Many of today's most sensitive datasets are captured in large graphs. Such
datasets can include maps of autonomous systems in the Internet, social networks
representing billions of friendships, or connected records of patent
citations. 
Controlling access to these datasets is a difficult challenge.  More
specifically, it is often the case that owners of large graph datasets would
like to share access to them to a fixed set of entities without the data
leaking into the public domain.  For example,  an ISP may be required to
share detailed network topology graphs with a third party networking
equipment vendor, with a strict agreement that access to these sensitive
graphs must be limited to authorized personnel only.  Similarly, a large
social network like Facebook or LinkedIn may choose to share portions of its
social graph data with trusted academic collaborators, but clearly want to
prevent their leakage into the broader research community.  

One option is to focus on building strong access control mechanisms to
prevent data leakage beyond authorized parties.  Yet in most scenarios,
including both examples above, data owners cannot restrict physical access to
the data, and have limited control once the data is shared with the trusted
collaborator.  It is also the case that no matter how well access control
systems are designed, they are never foolproof, and often fall prey to
attacks on the human element, {\em i.e.} social engineering.  Another option
is to modify portions of the data to reduce the impact of potential data
leakages.  This has the downside of making the data inherently noisy and
inaccurate, and still can be overcome by data reconstruction or
de-anonymization attacks using external input~\cite{deanonymize}.  Finally,
these schemes are hard to justify, in part because it is very difficult to
quantify the level of protection they provide.

In this work, we propose a new alternative in the form of {\em graph
  watermarks}.  Intuitively, watermarks are small, often imperceptible
changes to data that are difficult to remove, and serve to associate some
metadata to a particular dataset.  They are used successfully today to limit
data piracy by music vendors such as Apple and Walmart, who embed a user's
personal information into a music file at the time of
purchase/download~\cite{mp3water}.  Should the purchased music be leaked onto
music sharing networks, it is easy for Apple to track down the user who was
responsible for the leak.  In our context, graph watermarks work in a similar
way, by securely identifying a copy of a graph with its ``authorized user.''
Should a shared graph dataset be leaked and discovered later in
  public domains (on BitTorrent
perhaps), the data owner can extract watermark from the leaked copy and use
it as proof to seek damages against the collaborator responsible for the
leak.  While not a panacea, graph watermarks can provide additional level of
protection for data owners who want to or must share their data, and perhaps
encourage risk-averse data owners to share potentially sensitive graph data,
{\em e.g.}  encourage LinkedIn to share social graphs with academic
collaborators.

To be effective, a graph watermark system needs to provide several key
properties.  {\em First}, graph watermarks should be relatively small
compared to the graph dataset itself.  This has two direct consequences: the
watermark will be difficult to detect (and remove) by potential attackers,
and adding the watermark to the graph has minimal impact on the graph
structure and its utility. {\em Second}, watermarks should be difficult to
forge and should not occur naturally in graphs, ensuring that the presence of
a valid watermark can be securely associated with some user, {\em i.e.}
non-repudiation. {\em Third,} both the embedding and extraction of watermarks
should be efficient, even for extremely large graph datasets with billions of
nodes and edges.  {\em Finally,} our goal is to design a watermark system
that works in any application context involving graphs.  Therefore, we make no
assumptions about the presence of metadata.  Instead, our system must
function for ``barebones'' graphs, {\em i.e.} symmetric, unweighted graphs
with no node labels or edge weights.

In this paper, we present initial results of our efforts towards the design of a
scalable and robust graph watermark system.  Highlights of our work can
be organized into the following key contributions. 
\vspace{-0.05in}
\begin{packed_itemize}
\item  First, we identify the goals and requirements of a graph watermark
  system.  We also describe an initial design of a graph watermark system that
  efficiently embeds \fixhanx{watermarks} into and extracts watermarks out of large graphs.  Graph
  watermarks are uniquely generated based on \fixhanx{a user private key, a
    secure graph key}, and the
  graph they are applied to.  We describe constraints on its
  applicability, and identify examples of graphs where watermarks cannot
  achieve desirable levels of key properties such as uniqueness.
\item Second, we provide a strict proof of uniqueness of graph watermarks,
  showing that it is extremely difficult for attackers to forge
  watermarks. 
\item Third, we evaluate our watermarks in term of distortion,
  \fixhanx{false positive}, and efficiency on a wide variety of large graph datasets.
\item Fourth, we identify two attack models, describe additional features to
  boost robustness, and evaluate them under realistic conditions.
\end{packed_itemize}

To the best of our knowledge, our work is the first practical proposal for
applying watermarks to graph data.  We believe graph watermarks are a useful
tool suitable for a wide range of applications from tracking data leaks to
data authentication. Our work identifies the problem and defines an initial
groundwork, setting the stage for follow-up work to improve robustness
against a range of stronger attacks.

\section{Background and Related Work}

In this section, we provide background and related work on the graph privacy
problem and discuss the use of watermark techniques in applications such as
digital multimedia as well as graphs.

\para{Graph Privacy.} 
Graph privacy is a significant problem that has been magnified by the arrival
of large graphs containing sensitive data, {\em e.g.} Facebook social graphs
or mobile call graphs.  Recent studies~\cite{anonymization-www,deanonymize}
show that deanonymization attacks using external data can defeat
most common anonymization techniques. 

A variety of solutions have been proposed, ranging from anonymization tools
that defend against specific structural attacks, or more attack-agnostic
defenses.  To protect node- or edge-privacy against specific, known attacks,
techniques utilize variants of {\em k-anonymization} to produce structural
redundancy at the granularity of subgraphs, neighborhoods or single
nodes~\cite{liu2008identity,zhou2008neighborhood,hay2007anonymizing, zou2009k}. Alternatively, randomization provides privacy
  protection by randomly adding,
  deleting, or switching edges~\cite{hanhijarvi2009randomization,
    ying2008randomizing}. Others
partition the nodes and then describing the graph at the level of partitions
to avoid structural re-identification~\cite{hay2008resisting}.  Finally,
other solutions have taken a different approach, by producing model-driven
synthetic graphs that replicate key structural properties of the original
graph~\cite{modeling_www}. One extension of this work utilizes
differential privacy techniques to provide a tunable accuracy vs. privacy
tradeoff~\cite{ale2011differentially}.




The goals of our work are quite different from prior work on graph
anonymization, meant to protect data before its public release.  We are
concerned with scenarios where graph data is shared between its owner and
groups of trusted collaborators, {\em e.g.} third party network vendors
analyzing an ISP's network topology, or Facebook sharing a graph with a small
set of academic researchers.  The ideal goal in these scenarios is to ensure
the shared data does not leak into the wild.  Once data is shared with
collaborators, reliable tools that can track leaked data back to its source
serve as an excellent deterrent.  Watermarking techniques have addressed
similar problems in other contexts, and we briefly describe them here.

\para{Background on Digital Watermarks.}
Watermarking is the process of embedding specialized metadata into multimedia
content such as images or audio/video files~\cite{lee2001survey}. This
embedded {\em watermark} is later extracted from the file and used to identify
the source or owner of the content.  These systems include both an embedding
component and an extraction or recovery component. The embedding component
takes three inputs: a watermark, the original data, and a key. The watermark
is embedded into the data in a way that minimizes impact on the data, {\em
  i.e.} transparent letters overlaid on top of an image.  The key is used as
a parameter to change the way the watermark is embedded, usually corresponds
to a specific user, and is kept confidential by the data owner to prevent
unauthorized parties from recovering and modifying the watermark.  Extraction
takes as input the watermarked data, the key, and possibly a copy of the
original data. Extraction can directly produce the embedded watermark or a
confidence measure of whether it is present.

Significant work has been done in digital watermarking, particularly image
watermarking~\cite{walton1995authentication,macq1995cryptology,bender1996datahiding,ruanaidh1996phase,xia1997multiresolution}.
Image watermarking techniques can be classified into two classes based on
their working domains. The first class of watermarks is applied to the
original domain of the image, the spatial domain. Basic techniques include
modifying the least significant bits of each image pixel on the original
image to encode the watermark~\cite{walton1995authentication,macq1995cryptology,bender1996datahiding}. 
The second class applies watermarks to the transformed domain of the image,
{\em i.e.}  the frequency domain.  The original data is first transformed
into frequency domain using DCT~\cite{piva1997dct},
DFT~\cite{ruanaidh1996phase} or DWT~\cite{xia1997multiresolution}, added a
sequence of small noises to several invisible frequencies, and then the
result is transformed back into spatial domain as the watermarked image. The
sequence of noises is the watermark, and can be extracted by carrying out
the reverse process on the watermarked image. 

Watermark techniques are already widely used today to protect
  intellectual property.  Watermark techniques~\cite{vectormap02,
    vectormapspectral03} have been studied to protect the abuse of digital
  vector maps.  Like image watermarks, these techniques can be classified as
  spatial domain methods and transformed domain methods. Unlike image
  watermarks, the spatial domain methods embed watermarks by modifying vertex
  coordinates~\cite{vectormap02}, while the transformed domain methods tend
  to transform vector maps into a different frequency domain, such as the
  mesh-spectral domain~\cite{vectormapspectral03}. Watermarks have also been
used to protect software copyrights, by adding spurious execution paths in
the code that would not be triggered by normal
inputs~\cite{collberg1999software,venkatesan2001software}.  These execution
paths are embedded as extra control flows between blocks of code, and are
triggered (or extracted) by either locating the code, or running the program
with a special input that triggers the alternate execution paths. Moreover,
algorithms have been proposed for watermarking relational
datasets~\cite{agrawal2002databases,li2005fingerprinting,muhammad2012relational}.
Much of this has focused on modifying numeric attributes of relations,
relying on the primary key attribute as indicator of watermark locations,
assuming that the primary key attribute does not change. Finally,
  watermarks, in the form of minute changes, have been applied to
  protect circuit designs in the 
  semiconductor industry~\cite{graphcolorwt98, graphpartitionwt01}.

\begin{figure}[t]
\centering
\subfigure[Embedding]
{\epsfig{file=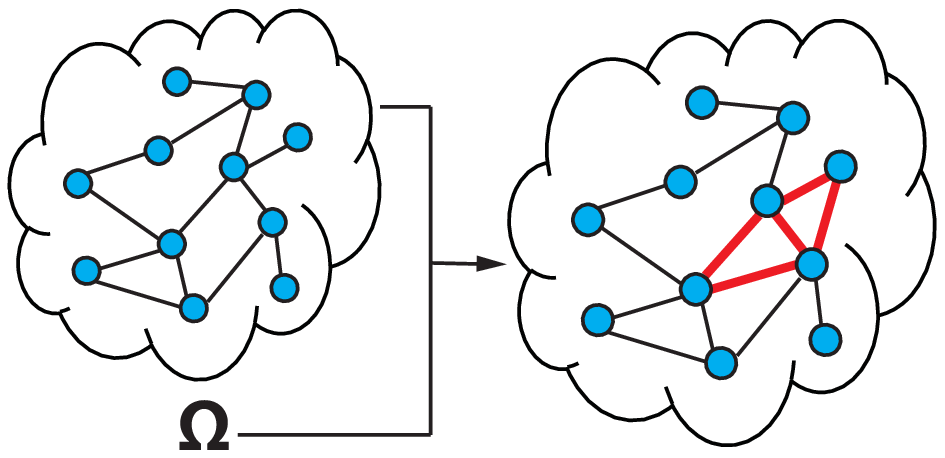,width=2.2in}}
\hspace{0.1in}
\subfigure[Extraction]
{\epsfig{file=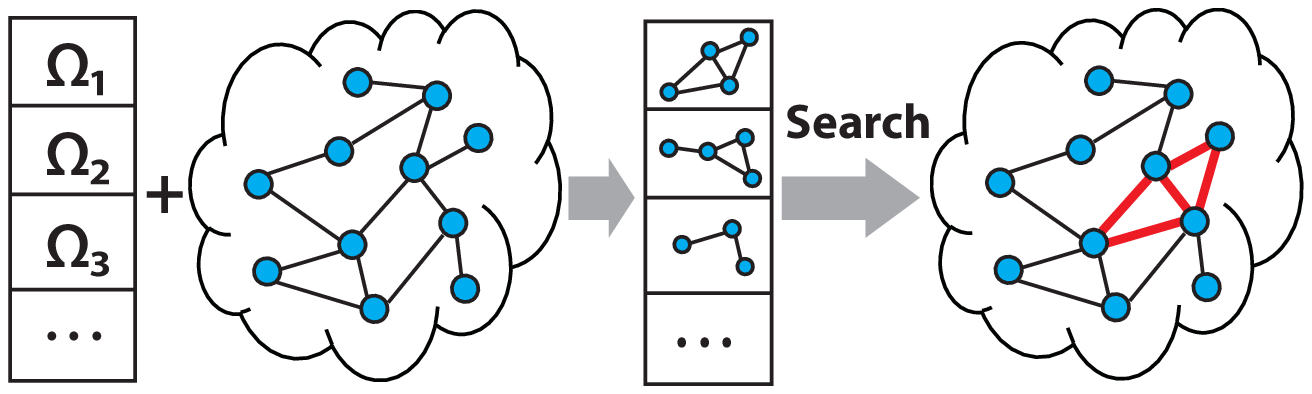,width=2.5in}}
\vspace{-0.1in}
\caption{Embedding and extracting graph watermarks. $\Omega$ is a secret
  random generator seed produced using the secure graph key and user's private key.}
\label{fig:watermark}
\end{figure}

\section{Goals and Attack Models}
\label{sec:goals}
To set the context for the design of our graph watermark system, we need
to first clearly define the attack models we target, and use them to guide
our design goals.

\para{Graph watermarks at a glance.} At a high level, we envision the graph
watermark process to be simple and lightweight, as pictured in
Figure~\ref{fig:watermark}.  Embedding a watermark involves {\em
    overlaying} the original graph dataset ($G$) with a small subgraph ($W$)
  generated using the original graph and a secret random generator seed \fixqing{($\Omega$)}.
Embedding the watermark simply means adding or deleting edges between
existing nodes in the original graph $G$, based on the watermark subgraph $W$.
Each authorized user \fixqing{$i$} receives only a watermarked graph customized for
  them, generated using a random seed \fixqing{$\Omega_i$} securely associated with \fixhan{her}.  The seed
  is generated through cooperation of her private key and a key securely
  associated with the original graph.

If and when the owner detects a leaked version of the dataset, the owner takes the
leaked graph, and ``extracts the watermark,'' by iteratively producing all
known watermark subgraphs $W_i$ associated with $G$ and each of the seeds
$\Omega_i$ associated with an authorized user.
The ``extraction'' process is actually a matching process where the data
owner can conclusively identify the source of the leaked data, by locating the
matching $W_i$ in the leaked graph.

In our model of potential attackers and threats,
we assume that attackers have access to the watermarked graph, but
not the original $G$.  Clearly, if an attacker is able to obtain the
unaltered $G$, then watermarks are no longer necessary or useful.

\para{Attack Models.}
The attackers' goal is to destroy or remove graph watermarks while
  preserving the original graph.  Watermarks are designed to protect the
  overall integrity of the graph data.  Thus we do not consider scenarios
  where the attackers sample the graph or distort it significantly in order
  to remove the watermark.  Doing so would be analogous to removing a portion
  of all pixels from a watermarked video, or applying a high pass frequency
  filter to watermarked music.  Under these constraints, we consider two
  practical attack models below.


\vspace{-.05in}
\begin{packed_itemize}
\item {\em Single Attacker Model.} For a single attacker with access to one
watermarked graph, it will be extremely difficult to detect the watermark
subgraph. Without the key associated with another user, forging a watermark
is also impractical.  Instead, their best attack is to disrupt any potential watermarks
by making modifications, {\em i.e.} add or delete nodes or edges.

\item{\em Collusion Attack Model.} If multiple attackers join their efforts,
  they can recover the orginal graph by comparing
  multiple watermarked graphs,  identifying the differences ({\em i.e.}
  watermarks), and removing them. 
\end{packed_itemize}


\para{Design Goals.}  These attack models help us define the key 
characteristics required for an effective graph watermarking system.
\vspace{-.05in}
\begin{packed_itemize}
\item {\em Low distortion.} The addition of watermarks should have a small
  impact on overall structure of the original graph.  This preserves the
  utility of the graph datasets.
\item {\em Robust to modifications.} Watermarks should be robust to
  modification attacks on watermarked graphs, {\em i.e.} watermarks should
  remain detectable and extractable with high probability, even after the
  graph has been modified by an attacker.
\item {\em Low false positives.}  It is extremely unlikely for our
  system to successfully identify a valid watermark $W_i$ in an unwatermarked
  graph or a graph watermarked by $W_j$ where $i\neq j$. \fixhanx{When we
    embed a single watermark (Section~\ref{sec:basic})}, we also refer to
  this property  as {\em watermark uniqueness}. 
\end{packed_itemize}

Within the constraints defined above, designing a graph
watermark system is quite challenging, for several reasons.  First, the
subgraph that represents the watermark must be relatively ``unique,'' {\em
  i.e.} it is highly unlikely to occur naturally, or intentionally through
forgery.  A second, contrasting goal is that the watermark should not change
the underlying graph significantly (low distortion), or be easily detected.
Walking the fine line between this and properties of ``uniqueness'' likely
means we have to restrict the set of graphs which can be watermarked, {\em
  i.e.} for some graphs, it will be impossible to find a hard to detect
watermark that does not occur easily in graphs.  Finally, since any leaked
graph can have all metadata stripped or modified, watermark embedding and
extraction algorithms must function without any labels or identifiers. Note
that the problem of subgraph matching is known to be
NP-complete~\cite{cook1971complexity}.

\section{Basic Watermark Design}
\label{sec:basic}
We now describe the basic design of our graph watermarking system. The basic
design seeks to embed and extract watermarks on graphs to achieve watermark
uniqueness while minimizing distortion on graph structure. 
Our design has two key components:

\begin{packed_itemize} \vspace{-0.05in}

  \item {\bf Watermark embedding}: The data owner holds a graph key
    $K^G$ associated with a graph $G$ known only to her. Each user $i$
    generates its public-private cryptographic key pair $<K^i_{pub},
    K^i_{priv}>$ through a standard public-key algorithm~\cite{publickey},
    where $K^i_{pub}$ is user $i$'s public key and $K^i_{priv}$ is its
    corresponding private key. To share the graph $G$ with user $i$, the
    system combines input from user $i$ digital signature $K^i_{priv}(T)$ and graph
    key $K^G$ to form a random generator seed $\Omega_i$, and use
    $\Omega_i$ to
    generate a watermark graph $W_i$ for graph $G$.  The system embeds $W_i$
    into $G$ by selecting and modifying a subgraph of $G$ that contains the
    same number of nodes as $W_i$.  The resulting graph $G^{W_i}$ 
    is given to user $i$ as the watermarked graph.

\item {\bf Watermark extraction}: To identify the watermark in $G'$, we use
  $\Omega_i$ to regenerate $W_i$ and then search for the existence of
  $W_i$ within $G'$, for each user $i$.
  \vspace{-0.02in}
\end{packed_itemize}

In this section, we focus on describing the detailed procedure of these two
components.  We present detailed analysis on
the two fundamental properties of graph watermarks, {\em i.e.\/} uniqueness
and detectability in Section~\ref{sec:analysis}.



\subsection{Watermark Embedding}
\label{subsec:embedding}

The most straightforward way to embed a watermark is to directly attach the
watermark graph to the original graph. That is, if $W_i$ represents the
watermark graph for user $i$, and $G$ represents the original graph to be
watermarked, the embedding treats $W_i$ as an independent graph, and adds new
edges to connect $W_i$ to $G$.  However, this approach has two
disadvantages. {\em First}, direct graph attachment makes it easy for
external attackers to identify and remove $W_i$ from $G$ without using
graph key $K^G$ and user $i$'s signature $K^i_{priv}(T)$.  New edges connecting $W_i$ and $G$ must be carefully chosen
to reduce the chance of detection, and this is a very challenging task.  {\em
  Second}, attaching a (structurally different) subgraph $W_i$ directly to a
graph $G$ introduces larger structural distortions.


Instead, we propose an alternative approach that embeds the watermark graph
``in-band.''  That is, the embedding process first selects $k$ nodes ($k$ is
the number of nodes in $W_i$) from $G$ and identifies $S$, the corresponding
subgraph of $G$ induced by these $k$ nodes. It then modifies $S$ using $W_i$
without affecting any other nodes in $G$.  Because the watermark graph $W_i$
is naturally connected with the rest of the graph, both the risk of
detection and amount of distortion induced on the original graph $G$ are
significantly lower than those of the direct attachment approach.

We now describe the details of ``in-band'' watermark embedding,
which consists of four steps: (1) generating random generator seed
$\Omega_i$ from user $i$'s signature $K^i_{priv}(T)$ and graph key $K^G$; (2) generating the watermark graph
$W_i$ from the seed $\Omega_i$; (3) selecting the placement of $W_i$ on
$G$ by picking $k$ nodes from $G$ and identifying 
the corresponding subgraph $S$ induced by these $k$ nodes; 
and (4) embedding $W_i$ into $G$ by modifying $S$ to match structure of 
$W_i$. 


\para{Step 1: Generating random generator seed $\Omega_i$.} To generate
an unforgetable watermarked graph, we generate a random generator seed
$\Omega_i$~\cite{gilbert1959random} using user $i$'s signature $K^i_{priv}(T)$
and graph key $K^G$. 

Suppose the system intends to generate a watermarked version of graph $G$ at
time $T$ to share with a specific user $i$.  We begin by first sending user
$i$ with the current timestamp $T$. User $i$ responds with its signature
$K^i_{priv}(T)$, by encrypting the timestamp with its private key $K^i_{priv}$.  Before proceeding further, 
we validate the result $K^i_{priv}(T)$ to ensure it is from $i$, by
decrypting it with user $i$'s public key $K^i_{pub}$.   If the timestamps match, we
combine the signature $K^i_{priv}(T)$ and the graph key $K^G$ to
form the seed of the random graph generator for user $i$, $\Omega_i$.
A mismatch may indicate that user $i$ is a potential malicious user.

Note that $\Omega_i$ cannot be formed alone by the data owner who only holds
the graph key $K^G$, or by user $i$ who only owns its private key
$K^i_{priv}$. Therefore, results computed using seed $\Omega_i$, including
the random graph $W_i$ generated (Step 2) and the choice of graph nodes to
mark (Step 3), cannot be derived independently by the data owner or
identified by user $i$. 

\para{Step 2: Generating the watermark graph $W_i$.} 
We generate $W_i$ as a random graph with edge probability
of $p$ and node count $k$ ($k<<n$ where $n$ is the number of nodes in
$G$). The random edge generator uses $\Omega_i$ as the
seed~\cite{gilbert1959random}. The $k$ nodes of $W_i$ are ordered as 
$\{v_1, v_2, ..., v_k\}$. 



The key factor in this step is choosing the node count $k$ and the
edge probability $p$. As we will show in Section~\ref{subsec:uniqueproof},
the two parameters must satisfy the following requirement to ensure watermark
uniqueness:
\vspace{-0.05in}
\begin{equation} 
k\geq (2+\delta)\log_q{n}\vspace{-0.05in}
\end{equation}
where $q=\frac{1}{\max{(p,1-p)}}$ and $\delta$ is a constant $>0$.
Furthermore, it is easy to prove that  
$p=\frac{1}{2}$ minimizes the node count $k$ and the average edge
count $p\cdot {k \choose 2}$ of the watermark graph $W_i$. Intuitively, using
a compact watermark graph not only reduces the amount of distortion to $G$,
but also improves its robustness against malicious attacks. Therefore, we
configure $p=\frac{1}{2}$ and therefore $k=(2+\delta)\log_2{n}$. 
This produces a reasonably sized watermark graph ($k<$100) even for extremely
large graphs, {\em e.g.} the complete Facebook social graph ($\sim$1 billion
nodes in 2014).


\para{Step 3: Selecting the watermark placement on graph $G$.}  Next,
we identify $k$ nodes from $G$ and its corresponding subgraph $S$ to embed
the watermark graph.  To ensure reliable extraction,  we must
choose these $k$ nodes carefully, meeting these two
requirements. {\em First}, using $\Omega_i$ generated in
Step 1, the $k$ nodes must be chosen deterministically and remain
distinguishable from  
the other nodes of $G$.  {\em Second}, the set of the $k$ nodes
chosen for different watermarks (or different $\Omega_i$
values) must be easily distinguishable from each other to reinforce watermark
uniqueness. 
Our biggest challenge in meeting these requirements is that we cannot use
node IDs to distinguish nodes from each other.  Node IDs or any type of
metadata can be easily altered or stripped by attackers before or after
leaking $G'$, thereby making extraction impossible.

We address this challenge by using local graph structure around each node as its
``label.''  Specifically, we define a {\em node structure description}
(NSD) as a distinguishable feature of each node.  A node $v$'s NSD is
represented by an array of $v$'s sorted neighbor degrees. For example, if
node $v$ has three neighbors with node degrees 2, 6, 4, respectively, then
$v$'s NSD label is ``2-4-6.''  We then hash $v$'s NSD label into a numerical
value using a secure one-way hash {\em e.g.} SHA-1~\cite{sha-1}, and refer to
the result as node $v$'s {\em NSDhash}.

Next, we use $\Omega_i$ as the seed to randomly generate $k$ hash values, and
use each as an index ({\em e.g.} using a mod function) to identify a node in
$G$.  It is possible that multiple nodes have the same NSDhash, {\em i.e.} a
collision.  If this happens, we resolve the collision by using $\Omega_i$
again as an index into a sorted list of these nodes with the same
NSDhash.  The nodes can be sorted by any deterministic order, {\em e.g.} node
IDs in the original graph. Note that this process is only required for
embedding (and not extraction), so any deterministic order chosen by the
graph owner will suffice.

At the end of this step, we obtain $k$ ordered nodes from $G$, $X=\{x_1, x_2,
..., x_k\}$, and the corresponding subgraph $S=G[X]$ induced by the node set
$X$ on $G$.


\para{Step 4: Embedding the watermark graph $W_i$ into graph $G$.}  In this
step, we embed the watermark graph $W_i$ by modifying the subgraph $S=G[X]$
to match $W_i$.  Specifically, we match each (ranked) node in $W_i$, $\{v_1,
v_2, ..., v_k\}$ with the corresponding node in $S$ (or $X$),
$\{x_1,x_2,...,x_k\}$, {\em i.e.\/} $f: W \rightarrow S, f(v_i)=x_i$. And
once the nodes are mapped, we then apply an XOR operation on each edge of the
two graphs. That is, \fixhan{we consider the connection between $(v_i,v_j)$
  or $(x_i, x_j)$
  as one bit, {\em i.e.} an edge between $(v_i,v_j)$ or $(x_i, x_j)$ means $1$ and no edge
  between $(v_i,v_j)$ or $(x_i, x_j)$ means $0$. }If an edge $(v_i,v_j)$ exists in $W_i$, we modify the
corresponding edge value in $S$ from $(x_i, x_j)$ to $(x_i, x_j) \oplus 1$;
and if no edge $(v_i, v_j)$ exists in $W_i$, we modify the edge value $(x_i,
x_j)$ to $(x_i, x_j) \oplus 0$.  When the above edge modification process
ends, we also explicitly create edges between nodes $x_i$ and $x_{i+1}$ to
maintain a connected subgraph.  As a result, we transfer the subgraph $S$
into $S^{W_i}$ with the watermark graph $W_i$ embedded.  The reason for
choosing the XOR operation is that it allows the same watermark to be
embedded in the graph multiple times (at multiple locations), thus reducing
the risk of the watermark being detected and destroyed by attacks such as
frequent subgraph mining. We will discuss this in more details in
Section~\ref{sec:advanced}.

At the end of this step, we obtain a watermarked graph $G^{W_i}$ for user
$i$.  Before we distribute it to user $i$, we anonymize $G^{W_i}$ by
completely (randomly) reassigning all node IDs. Such
anonymization not only helps to protect user privacy, but 
also minimizes the opportunity for colluding attackers with multiple
watermarked graphs to identify the embedded
watermark (see Section~\ref{sec:advanced}).

\subsection{Watermark Extraction}
\label{subsec:extraction}
The watermark {\em extraction} process determines if a watermark graph $W_i$ is
embedded in a target graph $G'$.  If so, then $G'$ is a legitimate copy
distributed to user $i$.   The extraction process faces two key
challenges. {\em First}, the target graph $G'$, likely a watermarked version
of the original graph $G$, can easily be modified by users/attackers during
the graph distribution process.  In particular, all node IDs can be
very different from that of the original $G$.  Thus extraction 
cannot rely on node IDs in $G'$. {\em Second}, identifying
whether a subgraph exists in a large graph is equivalent to a subgraph
matching problem, known to be NP-complete.  To handle massive graphs, we need a
computationally efficient algorithm.



Our design addresses these two challenges by leveraging knowledge on the
structure of the subgraph where the watermark was embedded. This eliminates
the dependency on node IDs while significantly reducing the
search space during the subgraph matching process.  We describe our
proposed design in detail below. 

\para{Step 1: Regenerating the watermark.}  The owner performs the
extraction, and has access to the original graph $G$, graph key
  $K^G$, and user's signature 
$K^i_{priv}(T)$. For each user $i$, we combine its signature
$K^i_{priv}(T)$ and graph key $K^G$ to generate the random generator seed
$\Omega_i$ for that user. Then, we follow step $2-4$ described in
Section~\ref{subsec:embedding} to regenerate the watermark graph $W_i$,
identify the $k$ ordered nodes from $G$ and their NSD labels, and finally the
modified subgraph $S^{W_i}$ that was placed on a ``clean'' version of the
watermarked graph $G^{W_i}$.


\para{Step 2:  Identifying candidate watermark nodes on $G'$.}  Given the $k$
nodes $X=\{x_1,x_2,...,x_k\}$ identified from the original graph $G$, in this
step we need to identify for each $x_j$, a set of candidate nodes on the
target graph $G'$ that can potentially become $x_j$.  We accomplish this by
identifying all the nodes on $G'$ whose NSD labels are the same of
$x_j$ in the ``clean'' version of the watermarked graph $G^{W_i}$. Since multiple nodes can have the same NSD label, this
process will very likely produce multiple candidates. To shrink the 
candidate list,  we examine the connectivity between candidate nodes of $X$
on $G'$ and compare it to that among $X$ on $G^{W_i}$. If two
nodes $x_m$ and $x_n$ are connected in $G^{W_i}$,  we prune their candidate node
lists by removing any candidate node of $x_m$ that has no edge with any
candidate node of $x_n$ on $G'$ and vice versa.  This pruning process
dramatically reduces the search space.  After this step, we obtain
for each $x_i$ the candidate node list $C_i$ on the target graph $G'$. 



\para{Step 3: Detecting watermark graph $S^{W_i}$ on $G'$.} Given the candidate node list of each
node in $X$, we now search for the existence of $S^{W_i}$ on the target graph
$G'$. For this we apply a recursive algorithm to enumerate and prune the 
combinations of the candidate sets, until we identify $S^{W_i}$ or
exhaust all the node candidates.  The detailed algorithm is listed in
Algorithm 1.  In this algorithm, we use a node list $Y$ to record the list of
nodes in $G'$ which we have already finalized as the corresponding nodes in
$S^{W_i}$,  {\em i.e.} $Y=\{y_1, y_2, ..., y_m\}$
($m\leq k$). When the process starts, $Y=\emptyset$, $m=0$.




\begin{algorithm}[t]
\small
\caption{Recursive Algorithm for Detecting $S^{W_i}$ on $G'$.}
\begin{algorithmic}[1]
\STATE \textbf{Function:} SubgraphDetection($G'$, $S^{W_i}$, $\{C_1, C_2,
..., C_k\}$, $Y$, $m$)
\STATE{\textbf{Input:} Graph $G'$, watermark graph $S^{W_i}$, candidate node list $C_i$ for each node $x_i$ in $X$, identified node list
  $Y=\{y_1, y_2, ..., y_m\}$ ($m < k$)}
\STATE{\textbf{Output:} Identified node list $Y$}
\FOR{each node $c \in C_{m+1}$}
         \IF{$c \not \in Y$ and each edge $(c, y_t)$ in $G'$ ($t=1..m$) is the same as the edge
           $(x_{m+1}, x_t)$ in $S^{W_i}$ ($t=1..m$)}
                   \STATE{$Y= Y \cup c$}
                   \STATE{ $m=m+1$ }
                    \IF{ $m==k$}
                                     \STATE{Return $Y$}
                     \ELSE
                                      \STATE{SubgraphDetection($G'$, $S^{W_i}$, $\{C_1,
                     C_2, ..., C_k\}$,$Y$, $m$)}
                   \ENDIF
            \STATE{$Y = Y \setminus c$}
            \STATE{$m = m-1$}
           \ENDIF
\ENDFOR
\STATE{Return $Y$}
\end{algorithmic}
\end{algorithm}

{\em Discussion.} The above design shows that our watermark extraction
algorithm simplifies the subgraph search problem by restricting it to a small
number of selected nodes from a graph, thus avoiding the NP-complete subgraph
matching problem.  Also note that we target real graphs with very
high levels of node heterogeneity, {\em e.g.} small-world, power-law or highly
clustered graphs, which are very far from the uniform, lattice-like graphs
that are the worst case scenarios for graph isomorphism.  In practice, our system
can efficiently extract watermarks from real, million-node graphs, and do so
in a few minutes on a single commodity server (Section~\ref{subsec:eff}).

\section{Fundamental Properties}
\label{sec:analysis}
Having described the basic watermark system, we now present detailed analysis
on its two fundamental properties: {\em watermark uniqueness} where each
watermark must be unique to the corresponding user, and {\em watermark
  detectability} where the presence of a watermark should not be easily  detectable by
external users without the knowledge of the seed $\Omega_i$
  associated with user $i$.

\begin{table*}[t]
\centering
\caption {Suitability of watermarking for 48 of today's network graphs,
determined by comparing their node degree distribution $[N_{min}(G),N_{max}(G)]$
  and $k$-node subgraph density $[D_{min}(k),D_{max}(k)]$ to those of the embedded
  watermark graphs. 35 out of these 48 graphs are  suitable for
  watermarking. }
\label{tab:graphs}
\resizebox{2.05\columnwidth}{!}{
\begin{tabular}{|c|c |c c c  |c | c c | c c | c |}
\hline
{Graph}&\multirow{2}{*} {Graph} & \multirow{2}{*} {\# of Nodes} & \multirow{2}{*} {\#
  of Edges} &\multirow{2}{*} {Avg. Deg.}  &\multirow{2}{*} {$k$} &
\multicolumn{2}{|c|} {Node Degree Criterion} & \multicolumn{2}{|c|} {$k$-node Subgraph
  Density Criterion} & \multirow{2}{*}{Suitability}\\
\cline{7-10}
{Category}&{} & {} & {}& {} & {}& $(k+1)/2$ & [$N_{min}(G),N_{max}(G)$] & Watermark
& [$D_{min}(k),D_{max}(k)$] & { }\\  
\hline
\multirow{3}{*}{Facebook}&Russia & 97,134 & 289,324 & 6.0  & 39 & 20 & [1, 748] & 390 & [45, 701] &{\bf Yes} \\
{}&L.A. & 603,834 & 7,676,486 & 25.4  & 45 & 23 &[1, 2141] & 517 & [44, 975]& {\bf Yes}\\
{}&London & 1,690,053 & 23,084,859 & 27.3 & 48 &24 & [1, 1483] & 588 &  [47,
1128] & {\bf Yes} \\
\hline
{}&Epinions (1)& 75,879&405,740&10.7 &38&19&[1,3044]&370&[47,649]& {\bf Yes} \\
{}&Slashdot (08/11/06) &77,360&507,833&13.1&38&19&[1, 2540]	&370&[38, 668]& {\bf Yes} \\
{}&Twitter &81,306&1,342,303&33.0&38&19&[1, 3383]&370&[44, 703]&{\bf Yes} \\
{Other}&Slashdot (09/02/16)& 81,867&497,672&12.2&38&19&[1, 2546]&370&[38, 669]&{\bf Yes}\\
{Social}&Slashdot (09/02/21)&82,140&500,481&12.2&38&19&[1, 2548]&370&[38, 669]& {\bf Yes}\\
{Networks}&Slashdot (09/02/22)&82,168&543,381&13.2&38	&19&[1, 2553]&370&[38, 673]&{\bf Yes} \\ 
{}&GPlus&107,614&12,238,285&227.5&39&20&[1, 20127]&389.5&[53, 741]&{\bf Yes} \\
{}&Epinions (2)& 131,828&711,496&10.8&40&20&[1, 3558]&409.5&[51, 780]&{\bf Yes}\\
{}&Youtube&1,134,890&2,987,624&5.3&47&24&[1, 28754]&563.5&[47, 815]&	{\bf Yes} \\
{}&Pokec & 1,632,803&	22,301,964&27.3&	48&24&[1, 14854]&587.5&[47, 979]&	{\bf Yes} \\
{}&Flickr & 1,715,255 & 15,555,041 & 18.1  & 48 &24 & [1, 27236] & 588 & [51, 1128] & {\bf Yes} \\
{}&Livejournal & 5,204,176 & 48,942,196 & 18.8  & 52 & 26 & [1, 15017]  & 689& [51, 1326] & {\bf Yes} \\ 
\hline
{Citation}& Patents& 23,133 & 93,468 & 8.1 & 34 & 17 & [1, 280] & 297 & [37, 373]& {\bf Yes} \\
{Networks}&ArXiv (Theo. Cit.)& 27,770&352,304&25.4&34&17&[1, 2468]&297&[36, 534]& {\bf Yes} \\
{}&ArXiv (Phy. Cit.)&34,546&420,899&24.4&35&18&[1, 846]&314.5&[36, 544]&{\bf Yes}\\
\hline
{}&ArXiv (Phy.)&12,008&118,505&19.7&32&16&[1, 491]&263.5&[45, 496]& {\bf Yes} \\
{Collaboration}&ArXiv (Astro)& 18,772&198,080&21.1&33&17&[1, 504]&280&[37, 528]& {\bf Yes} \\
{Networks}&DBLP&317,080&1,049,866&6.6&43&22&[1,343]&472.5&[43,903]& {\bf Yes} \\
{}&ArXiv (Condense)& 3,774,768 & 16,518,947 & 8.8 & 51 & 26 & [1, 793] & 663 & [50,1063] & {\bf Yes} \\
\hline
{Communication}&Email (Enron)&36,692&183,831&10.0&35&18&[1,1383]&314.5&[43,515]& {\bf Yes}\\
{Networks}&Email (Europe)&265,214&365,025&2.8&42&21&[1,7636]&451&[74,683]& {\bf Yes}\\
{}&Wiki & 2,394,385 & 4,659,565 & 3.9& 49 & 25 & [1, 100029]& 612 &[65, 1066] & {\bf Yes}\\
\hline
{}&Stanford&281,903&1,992,636&14.1&42&21&[1,38625]&451&[66,861]& {\bf Yes}\\
{Web}&NotreDame& 325,729&1,103,835 &6.8&43&22&[1,10721]&472.5&[60,903]& {\bf Yes}\\
{graphs}&BerkStan&685,230&6,649,470&19.4&45&23&[1,84230]&517&[79,990]& {\bf Yes}\\
{}&Google& 875,713 & 4,322,051 & 9.9 & 46  & 23 & [1, 6332] & 540 &[72, 1033] & {\bf Yes}\\
\hline
{Location based}&Brightkite&58,228&214,078&7.4&37&19&[1,1134]&351&[41,665]& {\bf Yes}\\
{OSNs}&Gowalla&196,591&950,327&9.7&41&21&[1,14730]&430&[44,723]& {\bf Yes}\\
\hline
{}&Oregon (1) &11,174&23,409&4.2&31&16&[1,2389]&247.5&[95,352]& {\bf Yes} \\ 
{AS}&Oregon(2) & 11,461&32,730&5.7 &32&16&[1,2432]&263.5&[79,476] & {\bf Yes} \\ 
{Graphs}&CAIDA & 26,475&53,381&4.0 &34&17&[1,2628]&297&[113,436] & {\bf Yes} \\ 
{}&Skitter & 1,696,415 &  11,095,298 & 13.1  & 48 & 24 & [1, 35455] & 588 & [52, 1128] & {\bf Yes} \\ 
\hline
\hline
{}&Gnutella (02/08/04) & 10,876&39,994&7.4&31&16&[1,103]&247.5&[30,80] &  No \\
{}&Gnutella (02/08/25)&22,687&54,705&4.8 &34&17&[1,66]&297&[0,0]& No \\
{P2P networks}&Gnutella (02/08/24)&26,518&65,369&4.9&34&17&[1,355]&297&[0,44]&  No \\
{}&Gnutella (02/08/30) &36,682&88,328&4.8 &35&18&[1,55]&314.5&[35,70]&  No \\
{}&Gnutella (02/08/31) & 62,586 & 147,892 & 4.7 & 37 & 19 & [1, 95]& 351  & [39,76] &  No \\
\hline
{}&Amazon (03/03/02) &262,111&899,792&6.9 &42&21&[1,420]&451&[88,132]& No\\
{Amazon}&Amazon (2012) &334,863&925,872&5.5&43&22&[1,549]&472.5&[0,0]& No\\
{Co-purchasing}&Amazon (03/03/12) &400,727&2,349,869&11.7 &43&22&[1,2747]&472.5&[52,285]& No\\
{Networks}&Amazon (03/06/01) & 403,394 & 2,443,408 & 12.1 & 43 & 22 & [1, 2752] & 473  &  [52, 333] & No\\
{}&Amazon (03/05/05) &410,236&2,439,437&11.9 &43&22&[1,2760]&472.5&[50,333]& No\\
\hline
{Road}&Pennsylvania&1,088,092&1,541,898&2.8&47&24&[1,9]&563.5&[0,0]& No\\
{Networks}&Texas& 1,379,917&1,921,660&2.8 &47&24&[1,12]&563.5&[0,0]&No\\
{}&California& 1,965,206 & 2,766,607 & 2.8 & 49 & 25 &  [1, 12] & 612 &  [0, 0]& No\\
\hline
\end{tabular}
}
\end{table*}

\subsection{Watermark Uniqueness}
\label{subsec:uniqueproof}

As a proof of ownership, each embedded watermark should be unique for its
user. That is, given the original graph $G$ and the seed $\Omega_i$
associated with user $i$,
 the embedded watermark graph $S^{W_i}$ should not be isomorphic
to any subgraph of $G^{W_j}$ ($i \neq j$) where $G^{W_j}$ is the watermarked
graph for user $j$.  At the same time, $S^{W_i}$ should not be isomorphic to  any
  subgraph of the original graph $G$.   In the following, we
show that with high probability, our proposed graph watermark system produces
unique watermarks for any graph $G$. 

\begin{theorem}
Given a graph $G$ with $n$ nodes, let $k \geq (2+\delta) \log_{2} n$ for a
positive constant $\delta >0$. We apply the following process to create a watermarked graph $G^{W_i}$ for user $i$: 

\begin{packed_itemize}\vspace{-0.05in}
\item We create $k$ nodes, $V=\{v_1, v_2, ..., v_k\}$, and generate a random
  graph $W_i$ on $V$ with an edge probability of $\frac{1}{2}$. 
\item We randomly select $k$ nodes, $X$ = $\{x_1, x_2, ..., x_k\}$ from
  $G$, and identify the subgraph corresponding to these $k$ nodes $S=G[X]$. 
\item Using $W_i$, we modify $S$ as follows: we first map each node $x_i$ in
  $X$ to a node $v_i$ in $V$. Let
  $e(u,v)=1$ denote an edge exists between node $u$ and $v$ and $e(u,v)=0$
  denote otherwise. We modify each $e(x_i,x_j)$ in $S$ to $e(x_i,x_j) \oplus
  e(v_i,v_j)$.  We then explicitly connect nodes $x_i$ and $x_{i+1}$, {\em
    i.e.\/} $e(x_i,x_{i+1})=1$. The resulting $S$ now becomes $S^{W_i}$, and
  the resulting $G$ becomes $G^{W_i}$. 
\vspace{-0.05in}
\end{packed_itemize}
Let $G^{W_l}$ denote a watermarked graph for user $l$ ($l\neq i$), built
using a different seed $\Omega_l$. 
 Then with low
probability,  any subgraph of $G^{W_l}$ or $G$ is isomorphic to $S^{W_i}$.
\label{the1}
\end{theorem} 


\begin{proof}
We first show that with low
probability,  any subgraph of $G^{W_l}$ is isomorphic to
$S^{W_i}$.  Let $Y=\{y_1, y_2, ...y_k\}$ be a set of ordered nodes in $G^{W_l}$,
where each $y_i$ maps to a node $x_i$ in $X$. We define an event
$\mathcal{E}_Y$ occurs if the subgraph $G^{W_l}[Y]$ is isomorphic to $G^{W_i}[X]$ or $S^{W_i}$. 
Then the event $\mathcal{E}$ representing the fact that there exists at least
one subgraph on $G^{W_l}$ that is isomorphic to $S^{W_i}$ is the union of
events $\mathcal{E}_Y$ on all possible $Y$, {\em i.e.\/} $\mathcal{E}=\cup_Y
\mathcal{E}_Y$. 


Next, we compute the probability of event $\mathcal{E}$ by those of individual
  event $\mathcal{E}_Y$. Specifically, we first show that the probability of
  an edge exists between node $x_i$ and $x_j$ ($j\neq i+1$) in $S^{W_i}=G^{W_i}[X]$ is
  $\frac{1}{2}$.  This is because each edge in the random graph $W_i$ is
  independently generated with probability $\frac{1}{2}$.  After performing
  the XOR operation between $W_i$ and $S$, the probability of an edge exists
  between $x_i$ and $x_j$  ($j \neq i+1$) on $S^{w_i}$ is $\frac{1}{2}\cdot p_{ij} + (1-p_{ij}) \cdot
  \frac{1}{2}=\frac{1}{2}$ where $p_{ij}$ is the probability that an edge
  exists between $x_i$ and $x_j$ on $S$.  Thus the result of XOR between $W_i$
    and $S$ is also a random graph, and its edge generation is
    independent of that in $G^{W_l}, l \neq i$.  Furthermore, it is easy to
    show that our design applies XOR
  operations on ${k \choose 2}-(k-1)$ node pairs on the $k$ nodes, and
  each node pair has an edge with a probability of
  $\frac{1}{2}$. Thus, the probability of a subgraph $G^{W_l}[Y]$ being
  isomorphic to $S^{W_i}$ is 
$P(\mathcal{E}_Y)=\frac{1}{2}^{{k
    \choose 2} -(k-1)}\cdot \beta$ where $\beta\leq 1$ is the probability
that every ($y_i$, $y_{i+1}$) pair in $G^{W_l}[Y]$ is connected. Thus
$P(\mathcal{E}_Y) \leq \frac{1}{2}^{{k
    \choose 2} -(k-1)}$.


Since $\mathcal{E}=\cup_Y{\mathcal{E}_Y}$ and there are less than $n^k$ possible sets of
$k$ ordered nodes in $G^{W_l}$, we use the Union Bound to compute the probability
of event $\mathcal{E}$ as follows:\vspace{-0.08in}
\begin{equation}
\small
\begin{split}
&P(\mathcal{E})<n^k \cdot P(\mathcal{E}_Y)\leq n^k \cdot \frac{1}{2}^{{k \choose 2} -(k-1)}\\
& = 2^{\frac{k^2}{2+\delta}} \cdot
\frac{1}{2}^{\frac{k^2-3k}{2}+1} =\frac{1}{2}^{\frac{\delta
  k^2}{2(2+\delta)}-\frac{3k}{2}+1}
\end{split}
\label{eq:dj}\vspace{-0.05in}
\end{equation}
The above equation shows that the probability $P(\mathcal{E})$ reduces
exponentially 
to $0$ as $k$ increases. 

Finally, we
can apply the same method to show
that with low
probability,  any subgraph of $G$ is isomorphic to
$S^{W_i}$. This is because 
  the XOR operations between $W_i$ and $S$ produce a random graph 
that is independent of $G$. This concludes our proof. 
\end{proof}

\subsection{Watermark Detectability}
\label{subsec:detectability}
In addition to providing uniqueness, a practical watermark design should also
offer low detectability, {\em i.e.\/} with low probability each watermark
gets identified by external users/attackers. This means that without knowing
the seed $\Omega_i$ associated with user $i$,
the embedded watermark graph $S^{W_i}$ should not be easily distinguishable from the rest
of the graph $G^{W_i}$.  Therefore, the detectability would depend
heavily on the topology of the original graph $G$, {\em i.e.\/} a watermark graph can be
well hidden inside a graph $G^{W_i}$ if its structural property is not too different
from that of $G$. 

In the following, we examine the detectability of watermarks in terms of {\em
  a graph's suitability for watermarking}.  This is because directly quantifying
the detectability is not only highly computational
expensive\footnote{Each embedded watermark graph is similar to a random graph with
  $\frac{1}{2}$ edge probability. Thus the detectability is low if certain
  subgraphs of $G$ are also random graphs with similar edge probabilities. Yet
  identifying these subgraphs (and the embedded watermark graph) on a large
  graph incurs significant computation overhead.}, but also lacks a proper
metric.  Instead, we  cross-compare the key structural properties of $S^{W_i}$
and $G$, and define $G$ as being 
suitable for watermarking if its structure properties are similar to that of
$S^{W_i}$, implying a low watermark detectability. 




\para{Suitability for Watermarking.}  
To evaluate a graph's suitability for watermarks,  we first study the key
structure property of the embedded watermark graph 
$S^{W_i}$. To guarantee watermark uniqueness and minimize distortion,  the
watermark graph $S^{W_i}$ needs to be a random graph with an edge probability
of $\frac{1}{2}$ (except for the fixed edges between $x_i,x_{i+1}$ node pairs), and include $k=(2+\delta)\log_2{n}$ nodes.  Thus its average node degree is at least
$(k+1)/2$ and its average graph density is $({k \choose
  2}+k-1)/2$.

\begin{table}[H]
\centering
\caption {Size and density of subgraph on nodes with degree $> (k+1)/2$ in
each graph. Size is the number of subgraph nodes, and density is quantified
as average edges each node having inside the subgraph. }
\label{tab:subgraphs}
\resizebox{1\columnwidth}{!}{
\begin{tabular}{|c|c c | c c | c |}
\hline
\multirow{2}{*} {Graph} & \multicolumn{2}{|c|} {Subgraph}  &
\multicolumn{2}{|c|} {Watermark Graph} & \multirow{2}{*}{Suitability}\\
\cline{2-5}
{} & {Node \#} & {Avg. Deg.}& {$k$} & {Avg. Deg.}&{}\\  
\hline
Russia & 4,794& 22.2 & 39 & 20.0 &{\bf Yes} \\
L.A.  &196,174& 49.2 & 45 & 23.0 & {\bf Yes}\\
London &562,075& 56.1 & 48 &24.5 & {\bf Yes} \\
\hline
Epinions (1)&7,083&68.7&38&19.5& {\bf Yes} \\
Slashdot (08/11/06) &9,908&53.4&38&19.5& {\bf Yes} \\
Twitter &34,014&60.5&38&19.5&{\bf Yes} \\
Slashdot (09/02/16)& 10,065&53.0&38&19.5&{\bf Yes}\\
Slashdot (09/02/21)&10,105&53.2&38&19.5& {\bf Yes}\\
Slashdot (09/02/22)&10,605&53.4&38&19.5&{\bf Yes} \\ 
GPlus&68,828&347.1&39&20.0 &{\bf Yes} \\
Epinions (2)&10,363&83.5&40&20.5 &{\bf Yes}\\
Youtube&31,720&	45.1&47&24.0 &{\bf Yes} \\
Pokec &564,001&	53.0&48&24.5&{\bf Yes} \\
Flickr &136,202& 174.5 & 48 &24.5 & {\bf Yes} \\
Livejournal  &945,567 & 57.5 & 52 & 26.5 & {\bf Yes} \\ 
\hline
Patents &2,370& 15.6 & 34 & 17.5 & {\bf Yes} \\
ArXiv (Theo. Cit.)&12,054&43.4&34&17.5& {\bf Yes} \\
ArXiv (Phy. Cit.)&14,785&37.9&35&18.0 & {\bf Yes}\\
\hline
ArXiv (Phy.)&2,860&62.5&32&16.5& {\bf Yes} \\
ArXiv (Astro)&6,536&42.9&33&17.0 & {\bf Yes} \\
DBLP&15,004&17.3&43&22.0& {\bf Yes} \\
ArXiv (Condense) &178,455& 16.0 & 51 & 26.0  & {\bf Yes} \\
\hline
Email (Enron)&3,481&48.2&35&18.0& {\bf Yes}\\
Email (Europe)&1,779&44.0&42&21.5& {\bf Yes}\\
Wiki Talk&21,253 & 83.1 & 49 & 25.0& {\bf Yes}\\
\hline
Stanford &35,600&42.1&42&21.5& {\bf Yes}\\
NotreDame & 16,831&38.7&43&22.0& {\bf Yes}\\
BerkStan&110,202&57.0&45&23.0& {\bf Yes}\\
Google&55,431& 14.8 & 46  & 23.5& {\bf Yes}\\
\hline
Brightkite&4,586&30.8&37&19.0& {\bf Yes}\\
Gowalla&17,946&39.3&41&21.0& {\bf Yes}\\
\hline
Oregon (1) &264&17.1&31&16.0& {\bf Yes} \\ 
Oregon(2) &579&31.0&32&16.5& {\bf Yes} \\ 
CAIDA &575&16.0&34&17.5& {\bf Yes} \\ 
Skitter  &146,601& 50.0 & 48 & 24.5  & {\bf Yes} \\ 
\hline
\hline
Gnutella (02/08/04) &796&5.2&31&16.0&  No \\
Gnutella (02/08/25)&499&2.0&34&17.5& No \\
Gnutella (02/08/24)&709&2.7&34&17.5&  No \\
Gnutella (02/08/30) &1,001&3.8&35&18.0&  No \\
Gnutella (02/08/31)&1,276& 3.6 & 37 & 19.0  &  No \\
\hline
Amazon (03/03/02) &3,727&2.8&42&21.5& No\\
Amazon (2012) &5,318&2.5&43&22.0& No\\
Amazon (03/03/12) &25,717&6.7&43&22.0&No\\
Amazon (03/06/01)  &28,081& 7.3 & 43 & 22.0  & No\\
Amazon (03/05/05) &28,044&7.5&43&22.0& No\\
\hline
Pennsylvania&0&0&47&24.0& No\\
Texas& 0&0&47&24.0&No\\
California&0 & 0 & 49 & 25.0& No\\
\hline
\end{tabular}
}
\end{table}


Given these properties of the embedded watermark, we note that
  watermark node degree and density can be higher than those of many
  real-world graphs, such as those listed in
  Table~\ref{tab:graphs}. Intuitively, to ensure low detectability of such a
  watermark graph, suitable graphs should include a set of nodes ($D$) which are
  difficult to distinguish from the watermark nodes in term of node degree
  and subgraph density. Specifically, a suitable graph dataset needs to contain a set
  of nodes $D$ with degree comparable or higher than the watermark graph node
  degree; and the density of the subgraph on $D$ is at least
  comparable to the watermark graph density. If these two properties hold, the
  embedded watermark graph cannot be easily distinguished from $D$ in the
  graph, and therefore cannot be detected by attackers.

To capture the above intuition, we define that a graph $G$ is
suitable for watermarking if its node degree and graph density satisfy the
following two criteria.  First, the minimum and maximum node degree of $G$, 
denoted as $N_{min}(G)$ and $N_{max}(G)$ respectively,  need to satisfy
$N_{min}(G) \leq (k+1)/2
\leq N_{max}(G)$. Second, across all $k$-node subgraphs of $G$ whose node
degree expectation is greater than $(k+1)/2$, the minimum and
maximum graph density need to satisfy $D_{min}(k) \leq ({k \choose
  2}+k-1)/2 \leq D_{max}(k)$. Together, these two criteria ensure that the
embedded watermark graph can be ``well hidden'' inside $G^{W_i}$. 

To compute $D_{min}(k)$ and
$D_{max}(k)$, we need to enumerate all possible subgraphs of $G$, which is 
computationally prohibitive for large graphs. Thus we apply a
sampling method to estimate them. To estimate $D_{max}(k)$, we
identify the subgraph with the highest density using a greedy search: 
starting from a randomly chosen node $v_1$ with degree $> (k+1)/2$, pick the 2nd node $v_2$ with degree $> (k+1)/2$ that is
connected to $v_1$, then the 3rd node $v_3$ with degree $> 
(k+1)/2$ who has the most number of edges to $v_1$ and $v_2$. This
greedy search stops until we find $k$ nodes.  We repeat the same process for all
the nodes with degree $> (k+1)/2$, creating multiple subgraphs from
which we calculate the density and pick the
highest one.    To estimate $D_{min}(k)$, we apply a similar process to
locate multiple subgraphs except
that for each subgraph we locate the next node $v_{i+1}$ randomly as long as its node degree $>
 (k+1)/2$  and it connects to at least one of the existing nodes
$\{v_1,...v_i\}$.


\para{Suitability of Real Graph Datasets.}  We wanted to understand
  how restrictive our suitability constraints were in the context of real
  graph datasets available today.  We consider 48 real network graphs ranging
  from $10K$ nodes, $39K$ edges to $5M$ nodes and $48M$ edges. These graphs
  represent vastly different types of networks and a wide range of structural
  topologies.  They include 3 social graphs generated from Facebook regional
  networks matching Russia, L.A., and London~\cite{interaction}. They include
  12 other graphs from online social networks, including
  Twitter~\cite{leskovec2012egonetworks},
  Youtube~\cite{yang2012definingcommunities},
  Google+~\cite{leskovec2012egonetworks}, Slovakia
  Pokec~\cite{takac2012dataanalysis}, Flickr~\cite{mislove2007measurement},
  Livejournal~\cite{mislove2007measurement}, 2 snapshots from
  Epinions~\cite{richardson2003trust}, and 4 snapshots from
  Slashdot~\cite{leskovec2009community}.  We also add 3 citation graphs
  from arXiv and U.S. Patents~\cite{leskovec2005graphs}, 4 graphs capturing
  collaborations in arXiv~\cite{leskovec2005graphs} and
  DBLP~\cite{yang2012definingcommunities}, 3 communication graphs generated
  from 2 Email networks~\cite{leskovec2007evolution, leskovec2009community}
  and Wiki Talk~\cite{leskovec2010predicting}, 4 web
  graphs~\cite{leskovec2008statistical, albert1999internet}, 2
  location-based online social graphs from Brightkite and
  Gowalla~\cite{cho2011friendship}, 5 snapshots of P2P file sharing graph
  from Gnutella~\cite{leskovec2007evolution}, 4 Internet Autonomous System (AS)
  maps~\cite{leskovec2005graphs}, 5 snapshots of Amazon co-purchasing
  networks~\cite{leskovec2007marketing, yang2012definingcommunities}, and 3
  U.S. road graphs~\cite{leskovec2008statistical}. The statistics of all 
  graphs are listed in Table~\ref{tab:graphs}.


  For all graphs, we use $\delta = 0.3$ to ensure a 99.999\% watermark
  uniqueness, and compute and list the corresponding value of $k$ (from
  Equation~\ref{eq:dj}) in Table~\ref{tab:graphs}. Next we
  list the two criteria in terms of $(k+1)/2$ vs. $[N_{min}(G), N_{max}(G)]$,
  and $({k \choose 2}+k-1)/2$ vs. $[D_{min}(k), D_{max}(k)]$.  If a graph
  satisfies both criteria, our analytical results will hold for any
  watermarks embedded on it.

  We can make two observations based on results from
  Table~\ref{tab:graphs}. {\em First}, 35 out of our 48 total graphs are
  suitable for watermarking.  Also note that graphs describing similar
  networks are consistent in their suitability.  For example, all 15 graphs
  from various online social networks are suitable for watermarks. {\em
    Second}, all 13 graphs unsuitable for watermarks come from only 3 kinds
  of networks, {\em i.e.} Amazon copurchasing networks, P2P networks, and
  Road networks. These results in each group are self consistent. These
  results support our assertion that our proposed watermarking mechanism is
  applicable to most of today's network graphs with low detection risk. In
  practice, the owner of a graph can apply the same mechanism to determine if
  her graph is suitable for our watermark scheme.

  To understand key properties determining whether a graph is suitable for
  watermarking, we measure various graph structrual properties, including
  average node degree, node degree distribution, clustering coefficient,
  average path length, and assortativity.  We also consider the size and density
  of subgraphs on nodes with degree \fixhan{more} than watermark minimum average
  degree $(k+1)/2$. Our measurement results show that the size and density of
  subgraphs on nodes with degree $> (k+1)/2$ are the most important
  properties to determine suitability. Here, the size of these subgraphs is
  the number of nodes in the subgraph, and the density of the subgraph is
  measured as the average edges each node has inside the subgraph, {\em i.e.}
  average degree inside the subgraph. As shown in Table~\ref{tab:subgraphs},
  unsuitable graphs do not have subgraphs with \fixhan{density to comparable
    to watermarks}, while 
   subgraphs with the desired density can be found in graphs deemed suitable.
  These results are consistent with our intuition on
  quantifying suitability of watermarks.

\para{Summary.} Since the average watermark subgraph has high node
degree and density, a graph suitable for watermarking must include a set
of nodes, whose degree and subgraph density are comparable or even higher
than watermark subgraphs. We propose two criteria targeting at
node degree and subgraph density respectively to
quantify whether a graph is suitable for watermarking. We collect a large set
of available graph datasets today, and find that 35 out of 48 real 
graphs are suitable for watermarking. 
This promising result indicates that
watermark technique can be applied on most of real networks with low
probability to be identified.

\section{More Robust Watermarks}
\label{sec:advanced}
Our basic design provides the fundamental building blocks of graph
watermarking with little consideration of external attacks.  In practice,
however, malicious users can seek to detect or destroy watermarked graphs.
Here, we first describe external attacks on watermarks, and then present
advanced features that defend against the attacks.  Note that these
improvement techniques aim to increase the cost of attacks rather than
disabling them completely. Finally, we re-evaluate
the watermark uniqueness of the advanced design.

\subsection{Attacks on Watermarks}
\label{subsec:attack}
\fixhan{As discussed earlier, our attack model includes attacks trying to
  destroy watermarks while preserving the topology of the original graph.
  Based on the number of attackers, attacks on watermarks fall under our two
  attack models: single attacker and multiple colluding attackers.}  With
access to only one watermarked graph, a single attacker can modify nodes
and/or edges in the graph to destroy watermarks.  With multiple watermarked
graphs, colluding attackers can perform more sophisticated attacks by
cross-comparing these graphs to detect or remove watermarks.




\para{Single Attacker Model.}
The naive edge attack is easiest to launch, and tries to disrupt the
watermark by randomly adding or removing edges on the watermarked graph.
For the attacker, there is a clear tradeoff between the severity of the
attack (number of edges or nodes modified), and the structural change or
distortion applied to the graph structure.

At first glance, this attack seems weak and unlikely to be a real threat.
The probability of the attacker modifying one edge or node in the
embedded watermark graph $W_i$ is extremely low, given the relatively small
size of $W_i$ compared to the graph.  As shown later, however, this
attack can be quite disruptive in practice.  By modifying a node $n_i$ or an
edge connected to $n_i$, the attack impacts all of $n_i$'s neighboring nodes,
since their NSD labels will be modified.  These NSD label changes, while
small, are enough to make locating nodes in the watermark graph very
difficult. This effect is exacerbated in social graphs that exhibit a small
world structure, since any change to a supernode's degree will impact a
disproportionately large portion of nodes in the graph.

One extreme of this attack is to leak patial watermarked graphs or merge
several graphs together. With high probability, it can destroy the embedded
watermarks, but will significantly distort the graph topologies to reduce
their usability.  Thus, we do not consider such scenarios in our study.

\para{Collusion Attacks.} By obtaining multiple watermarked graphs, an
attacker can compare these graphs to eliminate watermarks. Since we anonymize
each watermarked graph by randomly reassigning node IDs (see Section 4.1),
attackers cannot directly match individual nodes across graphs. To compare
multiple graphs, we apply the deanonymization methods proposed
in~\cite{deanonymize,deanonymize11}. Specifically, we first match 1000 highest
degree nodes between two graphs based on their degree and neighborhood
connectivities~\cite{deanonymize11}, and then start from these nodes to
find new mappings with the network structure and the previously mapped
nodes~\cite{deanonymize}. 

Using the deanonymization method, attackers can then build a "cleaned''
graph, where an edge exists if it exists in the majority of
the watermarked graphs.  Since embedded watermark graphs are likely embedded
at different locations on each graph, a majority vote approach effectively
removes the contributions from watermark subgraphs, leading to a graph that
closely approximates the original $G$.

\subsection{Improving Robustness against Attacks}
\label{subsec:defense}

The attacks discussed above can disrupt the watermark extraction process in
two ways. First, adding or deleting nodes/edges in $G'$ changes node degrees,
and therefore nodes' NSD labels, thereby disrupting the identification of
candidate nodes during the second step of the extraction process;  second, 
adding or deleting nodes/edges inside the embedded watermark graph $S^{W_i}$ can
change the structure of the watermark graph, making it difficult to identify
during the third step of the extraction process.  To defend
against these attacks, we must make the watermark extraction process more robust
against attack-induced artifacts on both node and graph structure.  To do so,
we propose four improvements over the basic extraction design in
Section~\ref{subsec:extraction}.


\para{Improvements \#1, \#2: Addressing changes to node neighborhoods.} 
Extracting a watermark involves searching through nodes in $G'$ by their NSD
labels.  By adding or deleting nodes/edges, attackers can effectively change
NSD labels across the graph.  To address this, we propose two changes to the
basic extraction design. {\em First}, we bucketize node degrees (with 
bucket size $B$) to reduce the sensitivity of a node's NSD label to its
neighbors' node degrees. For example, with $B=5$, a node with degree 
$9$ will stay in the same bucket even if one of its edges has been removed (reducing its node degree to $8$).  {\em Second}, when selecting a watermark
node's candidate node list, we replace the exact NSD label matching with the
approximate NSD label matching. That is, a match is found if the overlap
between two bucketized NSD labels exceeds a threshold $\theta$. For example,
with $\theta=50\%$, a node with bucketized NSD label ``1-2-3-4'' would match
a node with label ``1-2-3'' since the overlap is 75\% $>\theta$.

These changes clearly allow us to identify more candidate nodes for each
watermark node, thus improving robustness against small local modifications.
On the other hand, more candidate nodes lead to more computation during the
subgraph matching step, {\em i.e.\/} step 3 in the extraction process.  Such
expansion, however, does not affect watermark uniqueness and detectability, 
since they are unrelated to the size of candidate pools.

\para{Improvement \#3, \#4: Addressing changes to subgraph structure.}  Random
changes made to $G'$ by an attacker has some chance of directly impacting a
node or edge in the embedded watermark.  To address this, we propose two
techniques. {\em First}, we add redundancy to watermarks by embedding the
same watermark graph $W_i$ into $m$ disjoint subgraphs $S_1,S_2,... S_m$
from the original graph $G$.  This greatly increases the probability of the
owner locating at least one unmodified copy of $W_i$ during extraction, even
in the presence of attacks that make significant changes to nodes and edges
in $G'$.  Note that since we embed watermarks on disjoint subgraphs, this
does not affect watermark uniqueness \fixhanx{$1-P(\mathcal{E})$. While embedding $m$ watermarks will impact
  false positive, which is $1-(1-P(\mathcal{E}))^m$}.


{\em Second}, it is still possible that all the watermark graphs are
``destroyed'' by the attacker and there are no matches in the extraction
process. If this happens, we replace the exact subgraph matching in the
step 3 of the extraction process with the approximate subgraph matching. That
is, a subgraph matches the watermark graph if the amount of edge difference
between the two is less than a threshold $L$.  By relaxing the search
criteria used in step 3 of the extraction process, this technique allows us
to identify ``partially'' damaged watermarks, thus again improving robustness
against attacks.  However, it can also increase false positives in watermark
extraction, reducing watermark uniqueness.  We show later in this section
that the impact on watermark uniqueness can be tightly bounded by controlling $L$.


\para{Improvement \#5: Addressing Collusion Attacks.}  Recall that for
powerful attackers able to match graphs at an individual node level, they can
leverage majority votes across multiple watermarked graphs to remove
watermarks.  To defend against this, our insight is to embed
watermarks that have some portion of spatial overlap in the graph, such that
those components will survive majority votes over graphs.

We propose a {\em hierarchical} watermark embedding process to protect
watermark(s) against collusion attacks.  To build watermarked graphs for $M$
users,  we uniform-randomly divide these $M$
users into 2 groups ($a_1$ and $a_2$) and associate each group
with a
public-private
key pair $<K^{a_1}_{pub}, K^{a_1}_{priv}>$ or $<K^{a_2}_{pub},
K^{a_2}_{priv}>$, which is generated and held by the data owner.  We repeat this to produce another
group partition and randomly divide $M$ users into 2 groups ($b_1$ and $b_2$)
associated 
with group key pairs $<K^{b_1}_{pub}, K^{b_1}_{priv}>$ and $<K^{b_2}_{pub},
K^{b_2}_{priv}>$ separately.  After this step, each user is assigned to two groups. For
example, a user $i$ is assigned to groups $a_1$ and $b_2$. 

To prevent the
data owner or users from forging group assignments, we modify step 1 in
Section~\ref{subsec:embedding} to achieve an agreement on group assignments between the data owner and each user. More specifically, at
time $T$ when the data owner tends to share its graph with a user $i$
assigned to two groups, {\em e.g.} groups $a_1$ and $b_2$, the data owner
first send user $i$ three items: current timestamp $T$ and two group signatures $K^{a_1}_{priv}(T)$ and $K^{b_2}_{priv}(T)$. User $i$ then
validates the two group signatures using the two group public keys
$K^{a_1}_{pub}$ and $K^{b_2}_{pub}$. If the timestamps encrypted using group
private keys are $T$, user $i$ agrees the group assignment, saves the three
items, and  sends back its personal signature, {\em i.e.} $K^i_{priv}(T)$; otherwise, user $i$ rejects the
group assignment. Once the data owner receives user $i$'s signature
$K^i_{priv}(T)$, it validates this timestamp with user $i$'s public key. If
it is valid, the data owner generates three seeds for user $i$: $\Omega_i$ by
combining $K^i_{priv}(T)$ and $K^G$, $\Omega_{a_1}$ by combining
$K^{a_1}_{priv}$ and $K^G$, and $\Omega_{b_2}$ by combining
$K^{b_2}_{priv}$ and $K^G$, where $K^G$ is graph key for graph $G$. Through
this agreement scheme, either the data owner or users can not forge their
group assignments. Moreover, since the generated seed for each group
is unique,  we can make sure that only one unique watermark
corresponds to each group. 

To embed the watermarks for user $i$, we first follows step $2-4$ in
Section~\ref{subsec:embedding} to 
embed two {\em group watermarks} using its two group seeds generated through
the above method, {\em i.e.} $\Omega_{a_1}$ and $\Omega_{b_2}$ in the
example. We then use user $i$'s individual seed, {\em i.e.} $\Omega_i$, to embed
an {\em individual watermark}.  When generating the group watermarks, we make
sure that 1) the group watermark remains the same for users in the same
group; and 2) watermarks corresponding to different groups do not
overlap with each other, or with each user's individual watermark graph.
Note that because the group and individual watermarks are generated with
different seeds, this hierarchical embedding process does not affect watermark
uniqueness.

Under this design, a collusion attack can successfully destroy all the
watermarks (group or individual) only if the majority of the watermarked
graphs come from different user groups. Otherwise, the majority vote on raw
edges will preserve the ``group watermark.''  We can compute the success rate
of the attack by the following equation, which represents the probability
that the majority of the graphs obtained by the attacker come from different
user groups:
\begin{equation} \vspace{-0.05in}
\lambda(M_a, J)= \left (1-J \sum^{M_a}_{i=\lceil\frac{M_a+1}{2}\rceil}{M_a\choose i}\cdot
(\frac{1}{J})^i\cdot (\frac{J-1}{J})^{M_a-i}\right)^2  
\label{eq:cap}
\end{equation}
where $M_a$ is the number of watermarked graphs obtained by the attacker and
$J$ is the number of groups in each group partition. The above design chose
$J=2$ because it minimizes $\lambda(M_a,J), \forall M_a$.  Furthermore, when
$M_a$ is odd, $\lambda(M_a,2)=0$; and when $M_a$ is even, $\lambda(M_a,2)$ is
at most 0.25 when $M_a=2$.  Note that in equation (\ref{eq:cap}) the
operation $(.)^2$ is due to the fact that we group the users twice into two
different group classes: $a_1,a_2$ and $b_1,b_2$.  If we only perform the
group partition once ({\em e.g.\/} dividing the users into $a_1, a_2$), then
$\lambda(2,2)=0.5$. This means that in practice we can further reduce
$\lambda$ by performing multiple rounds of group division (2 in the above
design) and adding more group watermarks.

Note that group watermarks contain much less information than single user
watermarks.  In fact, the more robust a group watermark, the larger
granularity (and less precision) it will provide.  Our proposed solution is
to extend the system by using additional ``dimensions,'' {\em e.g.} go beyond
the two dimensions of $a$ and $b$ mentioned above.  Combining results from
multiple dimensions will quickly narrow down the set of potential users
responsible for the leak. However, since a colluding attack requires the
involvement of multiple leakers, even identifying a single leaker is
insufficient.  Developing a scheme to reliably detect multiple
(ideally all) colluding users is a topic for future work.

\subsection{Impact on Watermark Uniqueness}
\label{subsec:advanced-proof}
To improve the robustness of our watermark system,  we relax the subgraph
matching criteria from exact matching to approximate matching with at most
$L$ edge difference. Such relaxation does not affect watermark detectability
because it does not change the embedding
process. However, it may affect watermark uniqueness, which we will
analyze next.



Consider two watermarked graphs $G^{W_i}$ and $G^{W_j}$ that were independently generated
for user $i$ and $j$ following the three steps defined in
Theorem~\ref{the1}.  Let $S^{W_i}$ and $S^{W_j}$ represent  the embedded watermark graph in
$G^{W_i}$ and $G^{W_j}$, respectively.  To examine the watermark uniqueness,
we seek to compute the probability that a subgraph in
$G^{W_j}$ differs from $S^{W_i}$ by at most $L$ edges. 

Our analysis follows a similar structure of Theorem~\ref{the1}'s proof. Let
$\mathcal{E}_Y$ denote the event where a subgraph of $G^{W_j}$ built on $k$
nodes $Y=\{y_1,y_2,..., y_k\}$ only differs from $S^{W_i}$ by at most $L$
edges.  Our goal is to calculate the probability of the event
$\mathcal{E}=\cup_Y \mathcal{E}_Y$, which is the union on all combinations of
$k$ nodes.  To do so, we first compute the probability of individual $\mathcal{E}_Y$. 


As shown in Theorem~\ref{the1}, the edges between ${k \choose 2}-(k-1)$ node pairs in $S^{W_i}$
are generated randomly with
probability $\frac{1}{2}$ and are independent of $G^{W_j}$, while the rest
$k-1$ edges ($<x_l,x_{l+1}>, l=1...k-1)$ are fixed.   Thus we can show that the probability that a subgraph $G^{W_j}[Y]$ differs from
$S^{W_i}$ by $h$ edges is upper bounded by $\frac{1}{2}^{e-k+1} \cdot {e \choose h}
$ where $e={k \choose 2}$.  Therefore, we can derive the
probability of $\mathcal{E}_Y$ as $P(\mathcal{E}_Y) \leq \frac{1}{2}^{e-k+1} \cdot \sum^{L}_{h=0}{e \choose
  h}$. 
And consequently, we have 
\begin{equation}
P(\mathcal{E}) \leq n^k \cdot \frac{1}{2}^{e-k+1} \cdot \sum^{L}_{h=0}{e \choose
  h}
\label{eq:unique2}
\end{equation}
where $e={k \choose 2}$, $k= (2+\delta)log_2{n}$, and $n$ is the node count
of $G^{W_j}$.



Next, given the probability of uniqueness $1-P(\mathcal{E})$, we compute the
upper bound on $L$ to ensure $1-P(\mathcal{E}) \geq 0.99999$ for all the
graphs in Table~\ref{tab:graphs} except \fixhan{Road graphs, Amazon graphs
  and P2P network graphs}.  Again we set $\delta=0.3$. The result is listed
in Table~\ref{tab:bound}, where the maximum limit of $L$ varies
\fixhan{between 0 and 12}.  In general, the larger the graph, the higher the
upper bound on $L$.

\begin{table}[t]
\caption{Upper bound of $L$ for the 35 network graphs.}
\label{tab:bound}
\resizebox{1.\columnwidth}{!}{
\begin{tabular}{|c| c c c c c|}
\hline
\multirow{2}{*}{Graph} & \multirow{2}{*}{Oregon (1)} & \multirow{2}{*}{Oregon
  (2)} & \multirow{2}{*}{CAIDA} & Email  & arXiv  \\
& & & & (Enron) & (Theo. Cit.)\\
\hline
$L$ Bound & 0 & 1 & 1 & 1 & 1 \\
\hline 
\hline
\multirow{2}{*}{Graph} & arXiv & arXiv &
arXiv& \multirow{2}{*}{Patent} & Slashdot\\
& (Phy. Cit.)  & (Phy.) & (Astro) & & (08/11/06)\\
\hline
$L$ Bound  & 1 & 1 & 1 &2 & 3 \\
\hline 
\hline 
\multirow{2}{*}{Graph} &\multirow{2}{*}{Twitter} & Slashdot & Slashdot
&Slashdot & \multirow{2}{*}{Brightkite}  \\
 & & (09/02/16) & (09/02/21) & (09/02/22) & \\
\hline
$L$ Bound & 3 & 3 & 3 & 3 & 3 \\
\hline 
\hline
Graph & Russia & Epinions (1) & Google+ & Epinions (2) & Standford\\
$L$ Bound & 4 & 4 & 4 & 5 & 5 \\
\hline 
\hline
\multirow{2}{*}{Graph}& Email & \multirow{2}{*}{Gowalla} & \multirow{2}{*}{BerkStand} & \multirow{2}{*}{DBLP} & 
\multirow{2}{*}{NorteDame}\\
&  (Europe) &  & & &\\
\hline
$L$ Bound & 5 & 5 & 6 & 7 & 7 \\
\hline 
\hline
Graph & L.A. & London & Flickr & Wiki & Google\\
$L$ Bound & 8 & 8 & 8 & 8 & 8 \\
\hline 
\hline
\multirow{2}{*}{Graph} & \multirow{2}{*}{Skitter} & \multirow{2}{*}{Youtube}
& \multirow{2}{*}{Pokec} & arXiv & \multirow{2}{*}{Livejournal}\\
& & & & (Condense) & \\
\hline
$L$ Bound & 8 & 9 & 9 & 11 & 12 \\
\hline
\end{tabular}
}
\end{table}

\section{Experimental Evaluation}
\label{sec:eval}
We evaluate the proposed graph watermarking system using
real network graphs.  We consider three key performance metrics, {\em false positive},
 {\em graph distortion} and {\em watermark
  robustness}.  Having analytically quantified the watermark
uniqueness in Section~\ref{sec:analysis} and~\ref{sec:advanced}, we focus on
examining  graph distortion and watermark robustness while ensuring 
\fixhanx{false positive less than 0.001\%}. 
We also study the computational
efficiency of the proposed watermark embedding and extraction schemes. 

\para{Experiment Setup.}  Given the large number of graph computations per
data point, we focus our experiments on two of the larger network graphs
listed in Table~\ref{tab:graphs}, the LA regional Facebook graph and the
Flickr network graph.  The two graphs have very different sizes and graph
structures. 
\fixhanx{To guarantee less than $0.001\%$ false positive}, we select
$\delta=0.3$, and the $k$ values for the LA and Flickr graphs are 45 
and 48, respectively.  For our basic design, we generate 1 watermark per
graph. For our advanced design, we set $L$ to 8, the degree bucket
size to 10, and the NSD similarity threshold to $\theta=0.75$. For each user,
we embed 5 watermarks in its graph, 3 as 
individual watermarks and 2 as group watermarks.  We chose these settings
because they are intuitive and work well in practice. We leave the optimization of these
parameters to future work. 

In the following, we present our experiment results in terms of 1) amount of
distortion introduced to the original graph due to watermarking, 2)
robustness of the watermark against attacks, and 3) computational efficiency
of our watermarking design.



\begin{table}[t]
\centering
\caption{Percentage of modified nodes and edges after embedding $5$
  watermarks into a graph and impact on graph structure (dK-2 Deviation).} 
\label{tab:wtsize}
\small
\begin{tabular}{|c| c| c |c |}
\hline
Graph & Nodes (\%) & Edges (\%) & dK-2 Deviation \\
\hline
Watermarked LA & 0.037\% & 0.033\%  & 0.0008\\
\hline
Watermarked Flickr &0.014\% & 0.019\% & 0.0001 \\
\hline
\end{tabular}
\end{table}


\begin{table}[t]
\centering
\caption{Graph metrics are consistent w/ and w/o watermarks.}
\label{tab:metrics}
\resizebox{1.\columnwidth}{!}{
\begin{tabular}{|c| c| c| c| c| c |}
\hline
Graph & AS & Avg. CC & Avg. Deg & Avg. Path & Dia. \\
\hline
LA Original& 0.21 & 0.19 & 25.4 & 4.6 & 14 \\
Watermarked & 0.21 & 0.19 & 25.4 & 4.6 & 14\\
\hline
Flickr Original & -0.02   & 0.18& 18.1 & 5.3&  21 \\
Watermarked & -0.02 & 0.18 &  18.1&5.3 & 21\\
\hline
\end{tabular}}
\end{table}

\subsection{Graph Distortion from Watermarks}
\label{subsec:distortion-result}
We consider three types of metrics for measuring the graph distortion from watermarks. 
\begin{packed_itemize} 
\vspace{-0.05in}
\item {\em Modifications to the raw graph} -- We count the number of nodes and
  edges modified by embedding watermarks. Intuitively, more
  modifications to the graph introduce higher distortion. 
\item {\em Deviation in the dK-2 distribution} -- We also measure the Euclidean
  distance between the dK-2 series of the original graph and that of the
  watermarked graphs\footnote{\small The Euclidean distance between two dK-2 series
    $G1$ and $G2$
    is defined by $\frac{1}{D}\sqrt{\sum_{<d_1,d_2>} (e^{G1}_{<d1,d2>} -
    e^{G2}_{<d1,d2>} )^2}$ where $D$ is the number of $<d_1,d_2>$
    combinations or entries in the dK-2 series.}.  Larger
  deviation in dK-2 series implies higher distortion to the graph structure. 
\item {\em Graph metrics w/ and w/o watermarks} -- Finally we measure the
  commonly used graph metrics before and after the watermarking, including
  node degree distribution, assortativity (AS)~\cite{modeling_www}, clustering
  coefficient (CC)~\cite{modeling_www}, average path length and diameter. Any
  large deviation in any of these metrics indicates that the watermarked
  graph experienced large distortion. 
\end{packed_itemize} 
\vspace{-0.05in}

We have examined the distortion introduced by both the basic and advanced
designs. We only show the results of the advanced design because it adds more
watermarks and thus leads to higher distortion. For both LA and Flickr graphs, we generate 10 different watermarked graphs
(using 10 different random generator seeds) and present the average result across
these graphs.  Because computing average path
length and diameter on these two large graphs is highly computational
intensive, we randomly sample 1000 nodes and compute the average
path length and diameter among them (following the same approach taken by
prior works on social graph analysis~\cite{interaction}). 


\begin{figure*}[t]
\begin{minipage}{0.5\textwidth}
\subfigure[Robustness,  LA basic]
{\epsfig{file=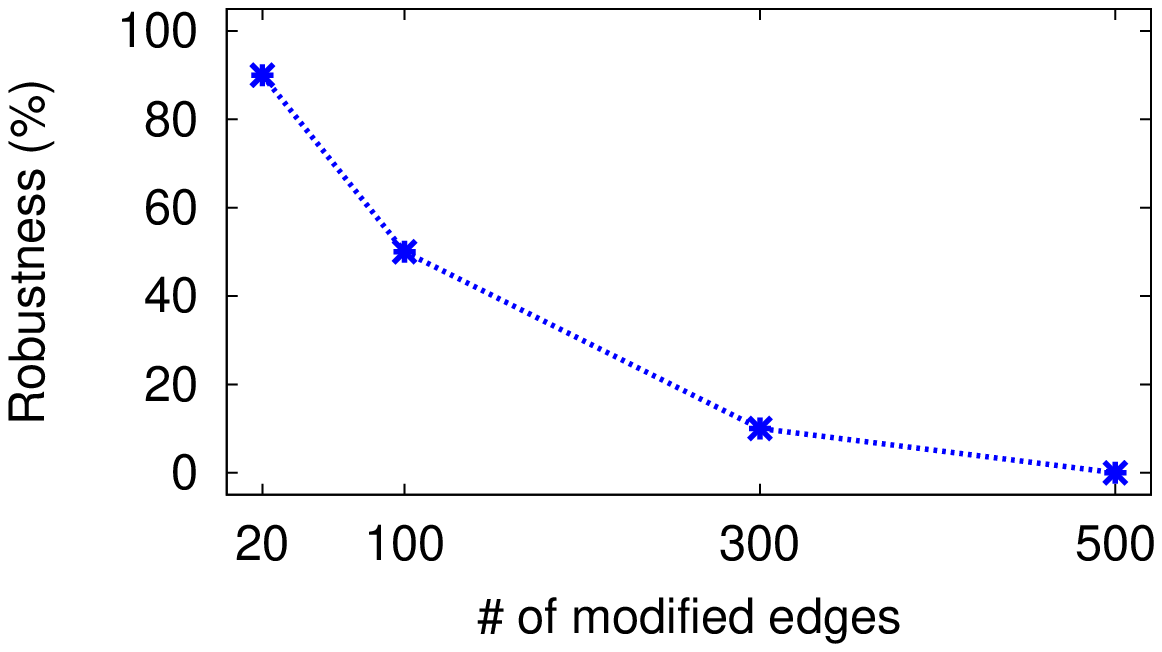,width=1.65in}}
\label{fig:larbas}
\hspace{-0.1in}
\subfigure[Robustness, Flickr basic]
{\epsfig{file=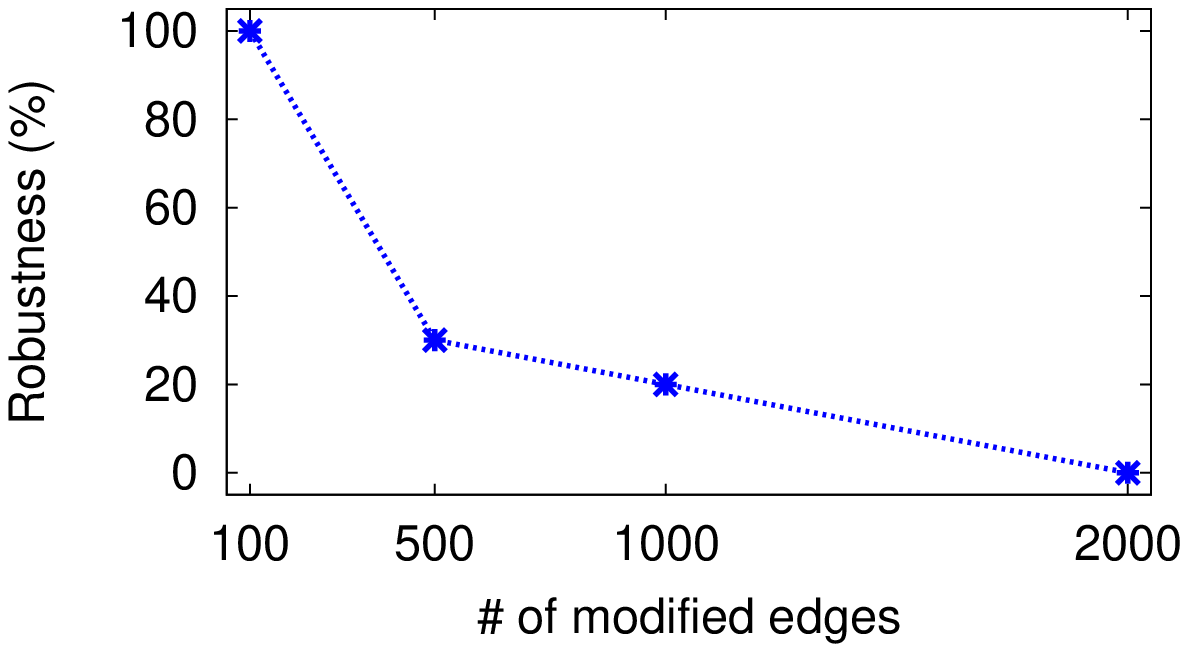,width=1.65in}}
\label{fig:flirbas}
\caption{The robustness of the basic design against the single attacker model.}
\label{fig:basicrob}
\end{minipage}
\hspace{0.1in}
\begin{minipage}{0.5\textwidth}
\subfigure[Distortion,  LA basic]
{\epsfig{file=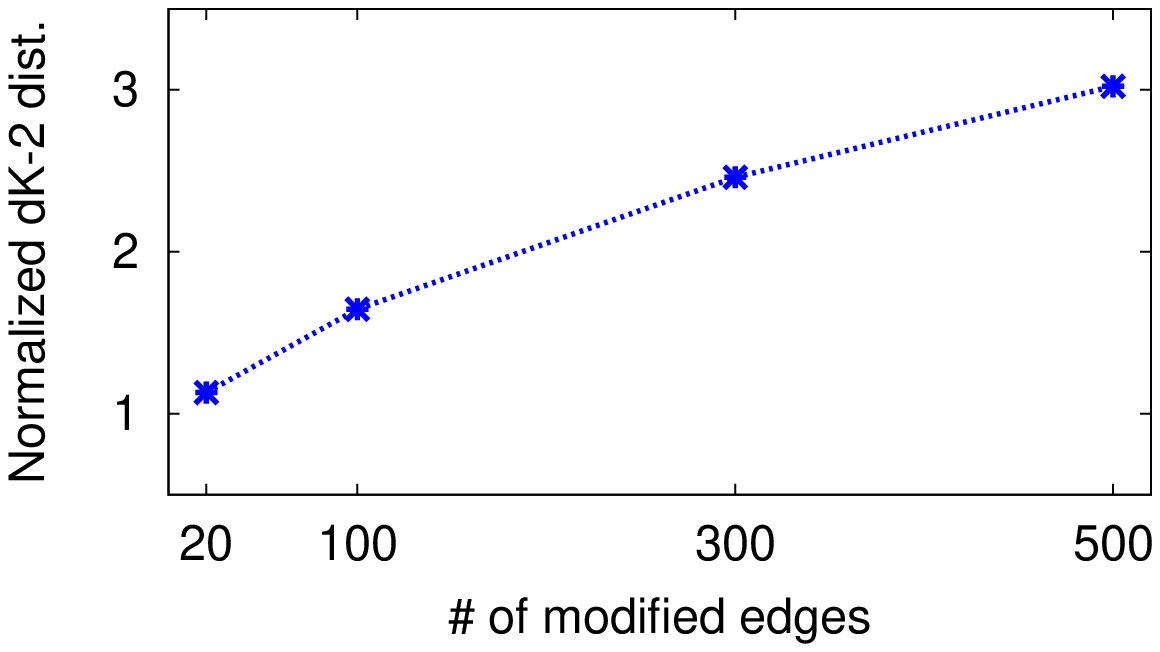,width=1.6in}}
\label{fig:ladkbas}
\subfigure[Distortion,  Flickr basic]
{\epsfig{file=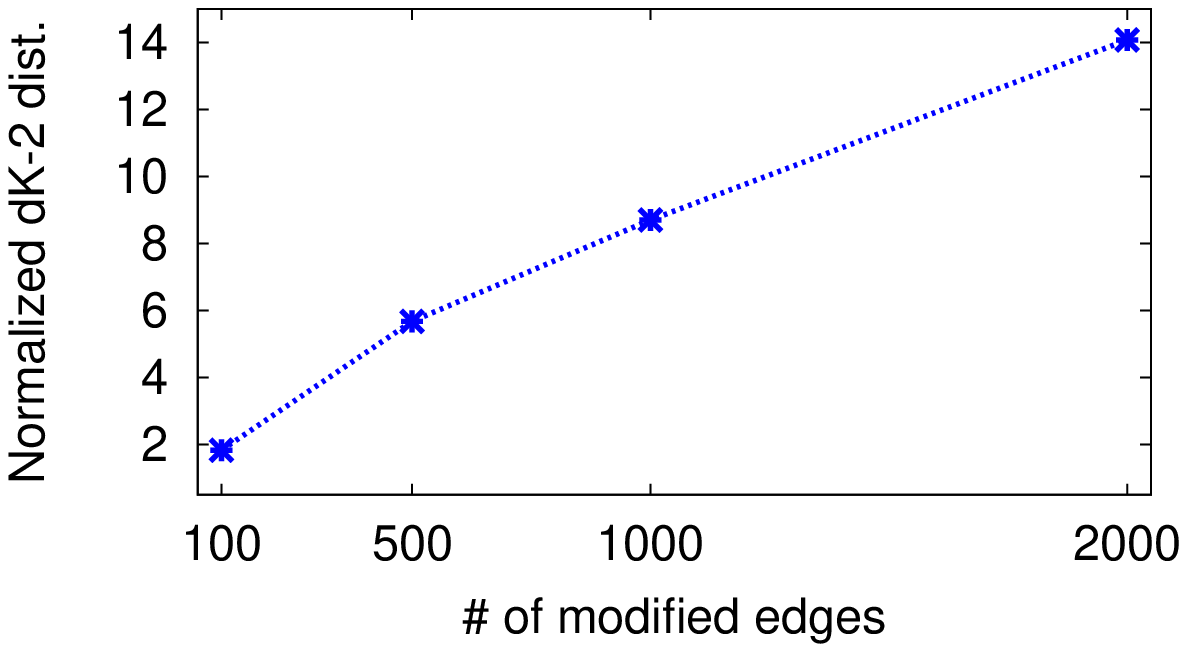,width=1.6in}}
\label{fig:flidkbas}
\caption{The distortion caused by the single attacker model in the basic design.}
\label{fig:basicdist}
\end{minipage}
\begin{minipage}{0.5\textwidth}
\subfigure[Robustness,  LA improve]
{\epsfig{file=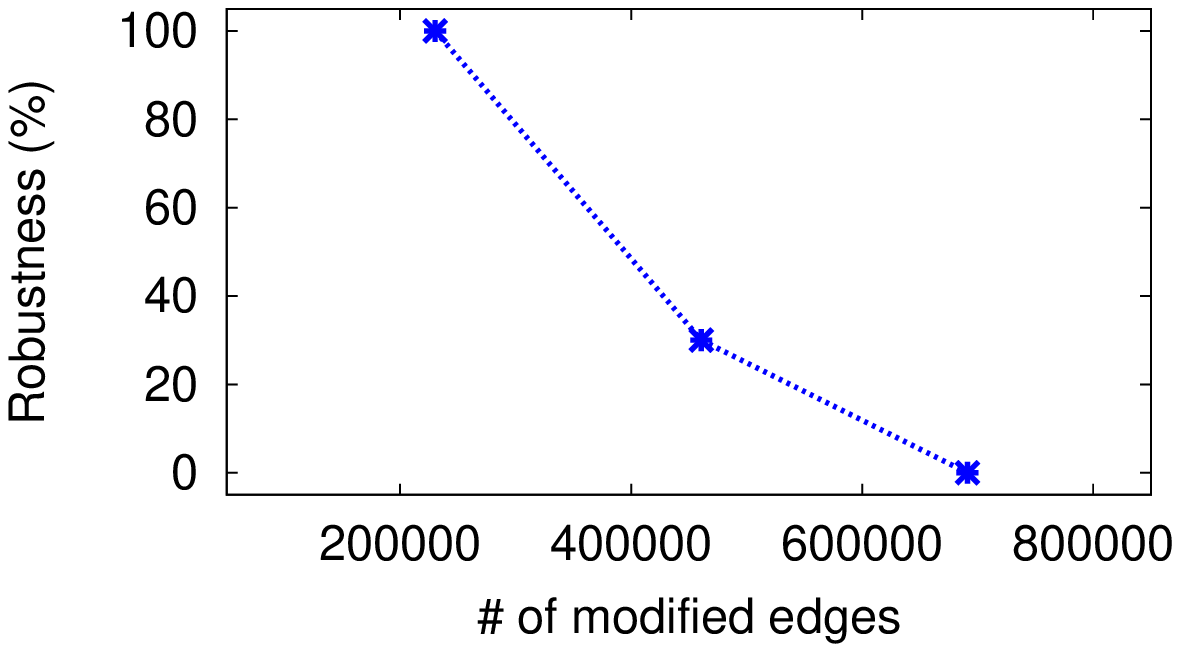,width=1.65in}}
\label{fig:larna}
\hspace{-0.1in}
\subfigure[Robustness, Flickr improve]
{\epsfig{file=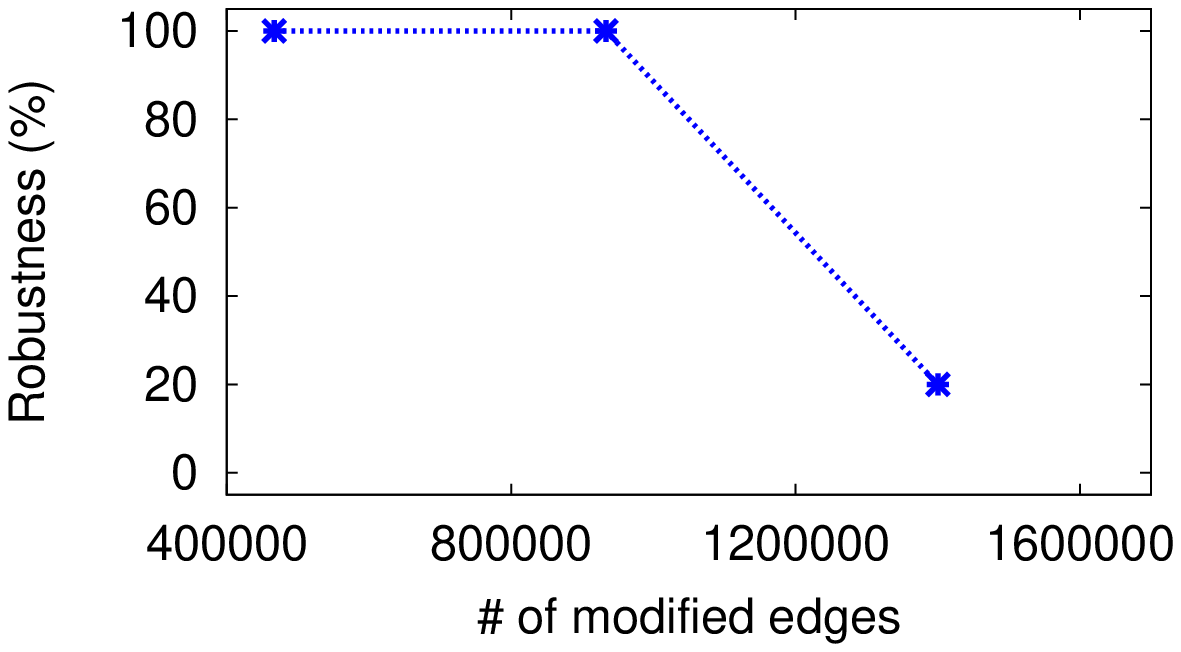,width=1.65in}}
\label{fig:flirna}
\caption{The robustness in the improved design against the single attacker model.}
\label{fig:narob}
\end{minipage}
\hspace{0.1in}
\begin{minipage}{0.5\textwidth}
\subfigure[Distortion,  LA improve]
{\epsfig{file=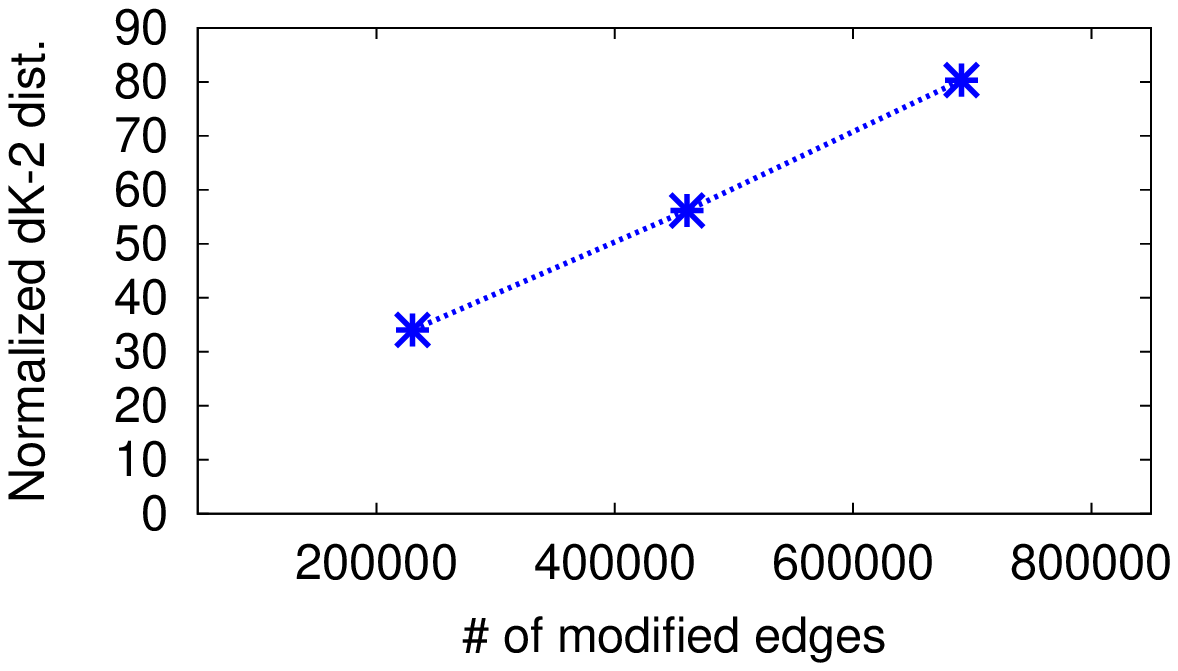,width=1.65in}}
\label{fig:ladkna}
\subfigure[Distortion,  Flickr improve]
{\epsfig{file=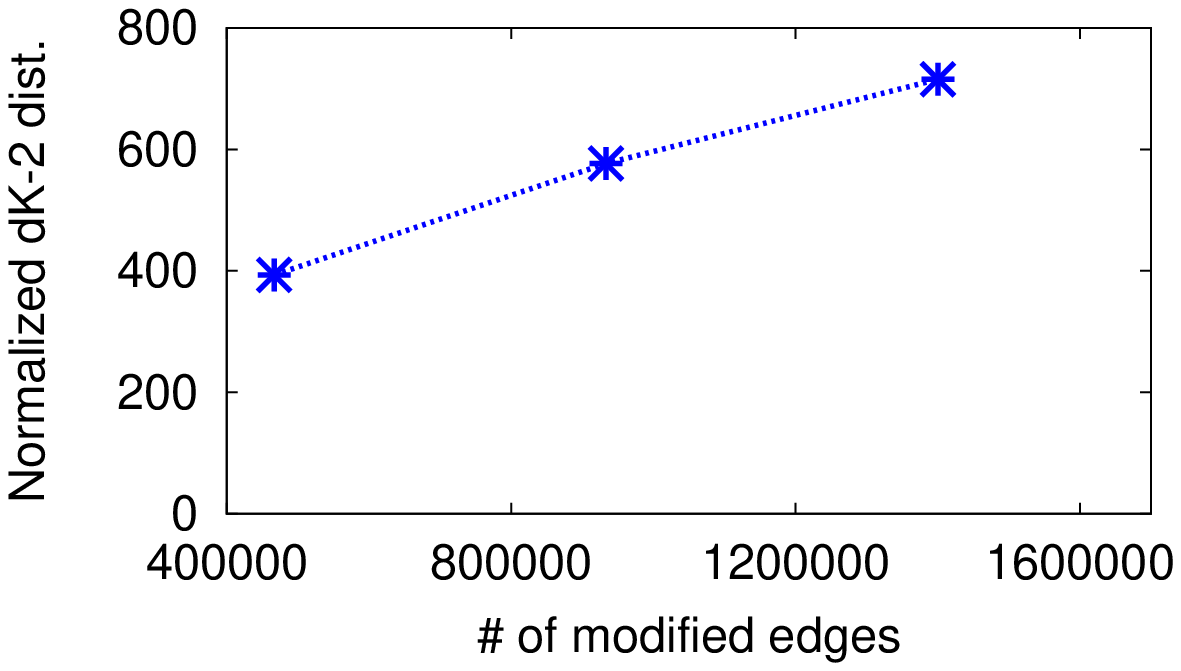,width=1.65in}}
\label{fig:flidkna}
\caption{The distortion caused by the single attacker model in the improved design.}
\label{fig:nadist}
\end{minipage}
\end{figure*}

 Table~\ref{tab:wtsize} shows
the percentage of modified nodes and edges by watermarking. Even after
embedding 5 watermarks, the modification is less than 0.04\% for LA and
0.02\% for Flickr.  These small changes imply little distortion on the
watermarked graphs. This is further confirmed by the average dK-2 distances
  for both graphs, 0.0008 for LA and 0.0001 for Flickr, indicating that the watermarked graphs are highly similar to the original
  graph. 

Table~\ref{tab:metrics} compares the original and watermarked graphs in terms
of  five representative graph metrics.  Similarly, for both LA and Flickr,
the graph metrics remain the same before and after watermarking.
We also examined the statistical
distribution of each metric and found no visible difference between the
graphs. 


Together, these results indicate that our proposed watermarking system
successfully embeds watermarks into graphs with negligible impact on graph
structure. This is unsurprising, given the extremely small size of watermarks
relative to the original graphs.  Thus we believe watermarked graphs can replace
the originals in graph applications and produce (near-)identical results.

\subsection{Robustness against Attacks}
\label{subsec:attackresult}
Next, we investigate how the proposed watermarking system performs in the
presence of attacks. For each of the two attack implementations discussed in
Section~\ref{subsec:attack}, we vary the attack strength and examine the
robustness against the attack as well as the cost of the attack.
Specifically, we repeat each experiment for 10 times, and examine two metrics:

\begin{packed_itemize}\vspace{-0.05in}
\item{\em Robustness} -- in the single attacker model, the robustness is
  quantified as the ratio of graphs from which we can successfully extract at least
one of the 3 individual watermarks. In the collusion attack, in addition
  to this ratio, we also measure the ratio of graphs where we can extract at least 
  one of the 5 watermarks (3 individual + 2 group watermarks). 
\item {\em Cost of the attack} -- the normalized distortion
produced on the attacked graph.  It represents the Euclidean distance between
the dK-2 series of the attacked graphs and that of the original graph, normalized by the Euclidean
distance between the dK-2 series of the clean watermarked graphs and that of
the original graph. 
 If the normalized distortion is larger than 1,  the attack introduces more distortion than embedding the
watermarks. \vspace{-0.05in}
\end{packed_itemize}

\begin{figure*}[t]
\begin{minipage}{0.5\textwidth}
\subfigure[Robustness, LA]
{\epsfig{file=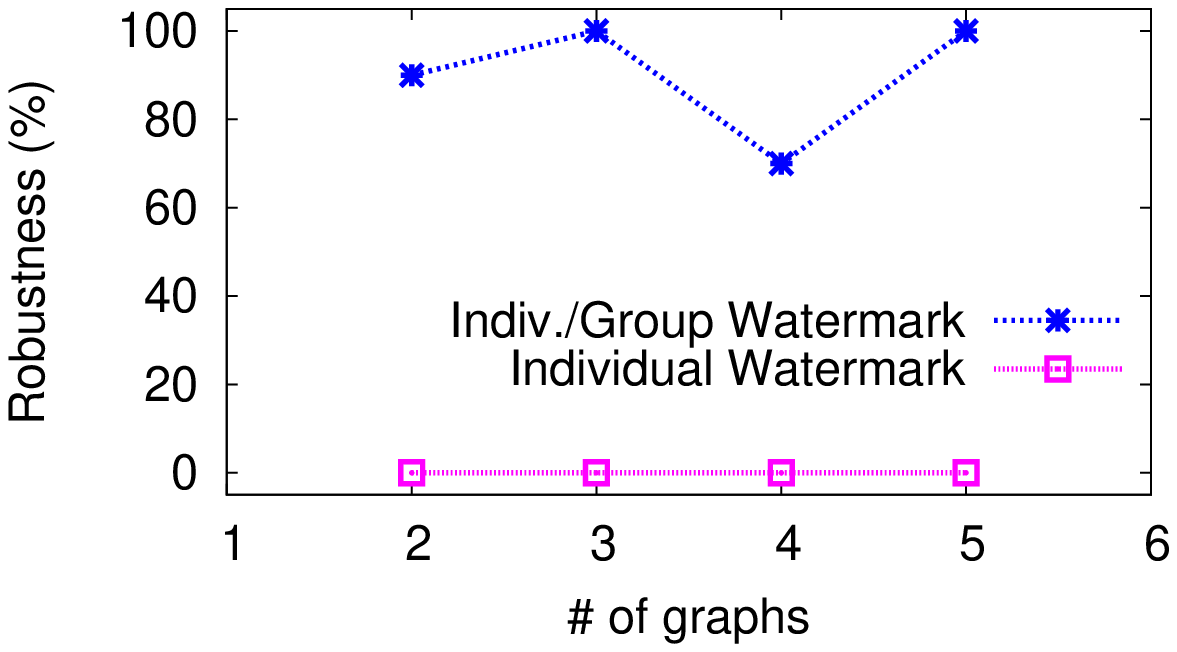,width=1.65in}}
\subfigure[Robustness, Flickr]
{\epsfig{file=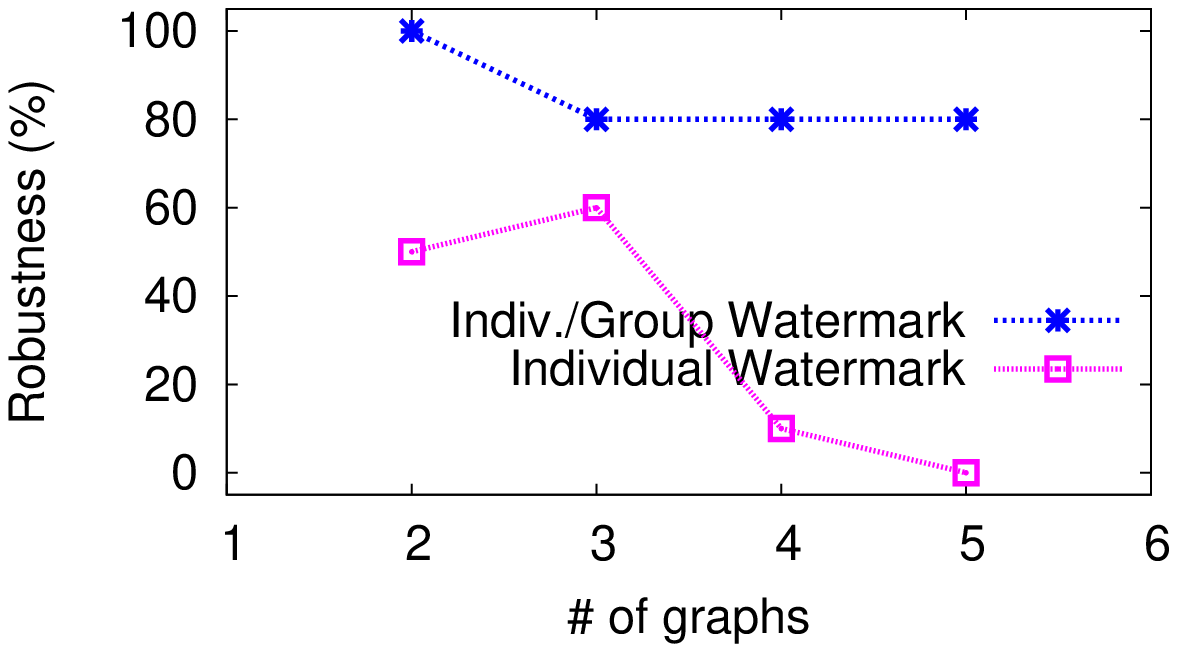,width=1.65in}}
\caption{The robustness against the collusion attack.}
\label{fig:collusionrob}
\end{minipage}
\begin{minipage}{0.5\textwidth}
\subfigure[Distortion, LA]
{\epsfig{file=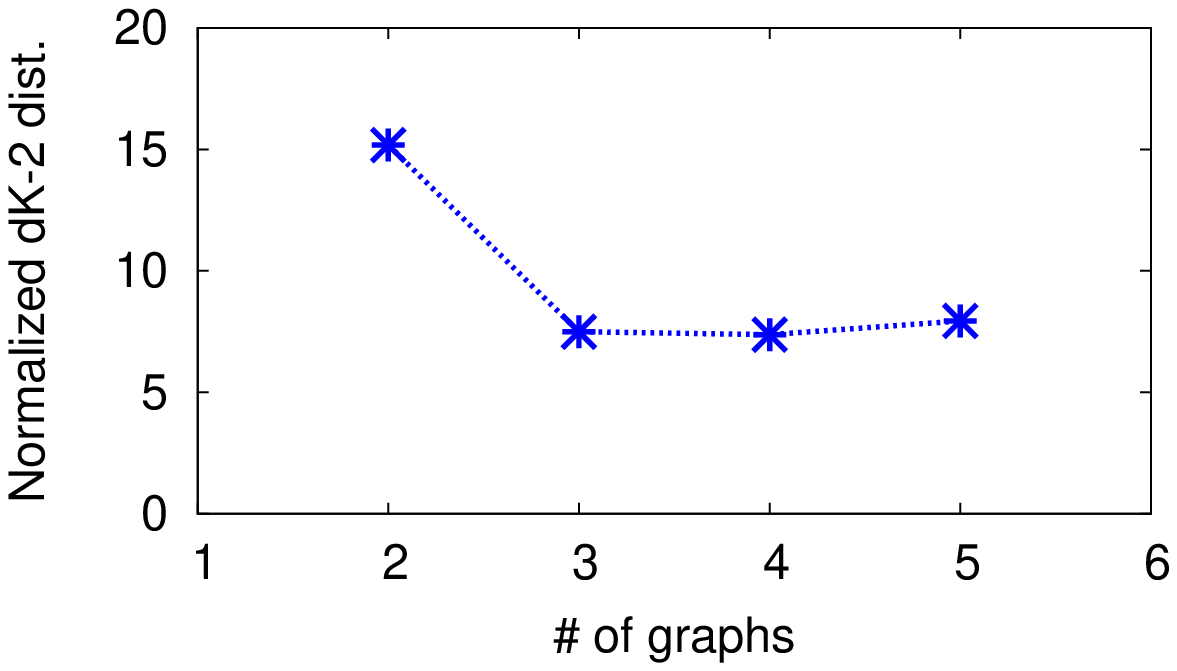,width=1.65in}}
\subfigure[Distortion, Flickr]
{\epsfig{file=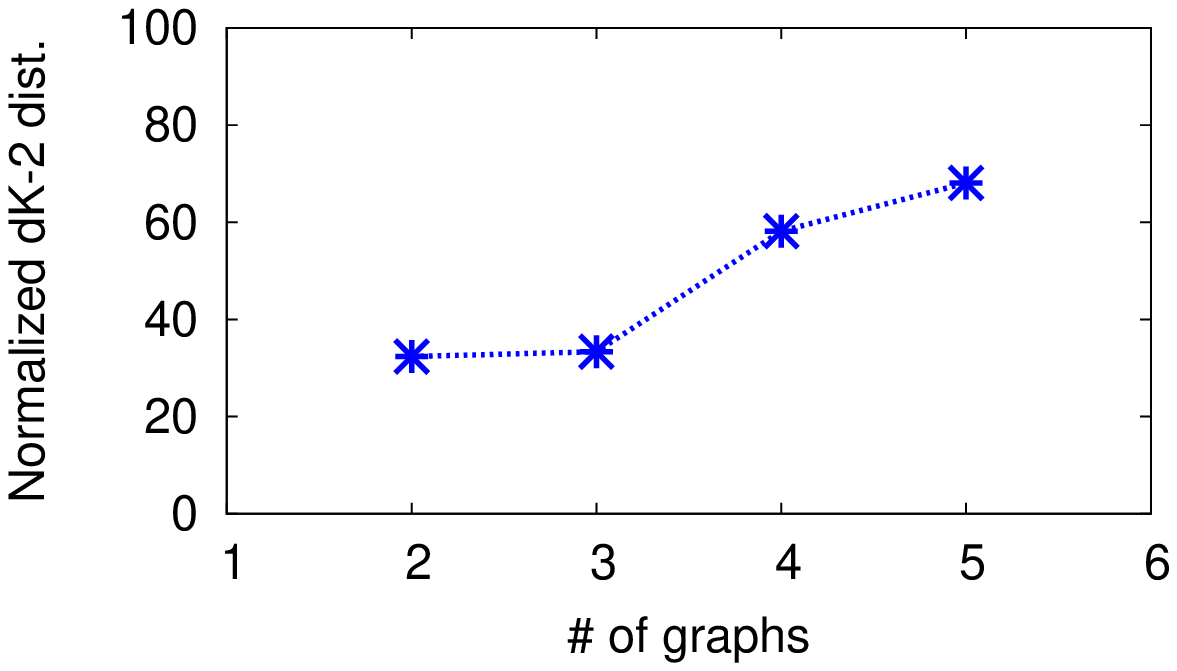,width=1.65in}}
\caption{The distortion caused by the collusion attack.} 
\label{fig:collusiondist}
\end{minipage}
\end{figure*}

\para{Results on the Single Attacker Model.}  For the single attacker model, we quantify
the attack strength by the number of modified edges. The robustness and the
cost of the attack are measured as a function of the number of modified edges.


To show how robustness is improved using the improvement mechanisms,
  we first evaluate the robustness results in the basic watermark
  scheme.  We run the single attacker model on the watermarked graphs by varying the number of modified
edge number, and repeat the experiment 10 times. The robustness here is
  quantified as the ratio of graphs from which we can successfully extract
  the watermark. 

Figure~\ref{fig:basicrob} shows the robustness of the basic
  watermark scheme against the single attacker model. It shows that randomly
  modifying a small number of edges
  disrupted the watermark subgraph extraction process. In LA, our basic design cannot recover the watermark with $100\%$ probability even when we modify $20$
  edges. In Flickr, a large graph, the robustness of the basic scheme reduces
  to less than $40\%$ when only $500$ edges are modified. In each case, at least one of the nodes in the
watermark subgraph had a modified NSD label (one of its neighbors' node
degree changed), and it could not be located in the extraction process. We
also look at the distortion caused by the attack shown in
Figure~\ref{fig:basicdist}. As expected, the small number of modified edges
causes small distortions in graph structures. For example, in LA, when the
robustness is $0$, the
distortion is around 3x more than that caused by embedding the
watermark. Both results show that watermarked graphs generated by the basic
scheme are easily disrupted by even small, single user attacks. 


Figure~\ref{fig:narob}(a)-(b) plot the robustness of watermarked LA and
Flickr graphs generated by the scheme with the improvement
  mechanisms against the single attacker model.  As expected,  the robustness decreases with the attack strength since more edges are modified to
``destroy'' watermarks. For LA, our system maintains 100\% robustness
up to 230K modified edges, which is around $400$x stronger than the
maximum attack strength handled in the graph generated by the basic design.
For Flickr, the system can
handle attack strength up to 933K modified edges, which is $>400$x stronger
 than the maximum attack strength in the basic design.
 This is because Flickr
is larger in size while having a similar watermark graph size $k$,  so the
attacker must modify more edges to destroy watermarks.   On the other hand,
results in Figure~\ref{fig:nadist} show that the cost of these
attacks is large. For Flickr, with more than 1.4M modified edges, an attack leads to 800x more
distortion over that caused by embedding the watermarks.  Together, these
results show that our watermark system with the improvement mechanisms is highly robust against
single user attacks. 

\para{Results on Collusion Attacks. }  To implement the collusion attack
desbribed in Section~\ref{subsec:attack}, we first generate 10 watermarked graphs
and randomly pick $M_a$ graphs from them as the graphs acquired by the
attacker. We vary the number of graphs obtained by the attacker $M_a$ between 2 to
5. For each $M_a$ value we repeat the experiments 10 times and report the
average value. Since watermarks generated by the basic design
  can be easily detected and removed by the powerful collusion attack,
here we focus on evaluating the robustness of the improvement mechanisms.

Figure~\ref{fig:collusionrob}(a)-(b) shows the robustness of the watermarked
LA and Flickr graphs against the collusion
attack. Figure~\ref{fig:collusionrob}(a) shows that in LA, by applying
majority votes on raw edges, the collusion attack can effectively remove all
3 individual watermarks. However, the attack is ineffective in removing both
group watermarks such that we can extract at least one group watermark in
more than 60\% of the attacked graphs. Here the
robustness values, deviate slightly from that projected by Equation
(\ref{eq:cap}) because we limit the number of
statistical sampling to 10 runs. Unlike LA, Figure~\ref{fig:collusionrob}(b)
plots that the collusion attack cannot remove all the individual watermarks
in Flickr when using 2 or 3 watermarked graphs. This is all because the
deanonymization method causes a large portion of nodes mismatched in Flickr (~30\%
nodes). Finally, Figure~\ref{fig:collusiondist} shows that the collusion
attack also introduce larger distortions in graph structure. This mainly
comes from the mismatch of the deanonymization methods. 

These results show that even a powerful collusion attack is ineffective in removing all
the embedded watermarks. Moreover, the potential inaccuracy of the deanonymization
method causes the collusion attack even weaker in removing individual
watermarks.  Of course, the attackers will eventually succeed in
disrupting watermarks if they are willing to modify larger portions of the
graph, thus sacrificing the utility of the graph.  While our work provides a
robust defense against attackers with relatively low level of tolerance for
graph distortion, we hope follow-on work will develop more robust defenses
against higher distortion attacks.

\subsection{Computational Efficiency}
\label{subsec:eff}
Here, we measure the efficiency of the watermarking system. There are two
components in the watermarking system, {\em i.e.} watermark embedding and
watermark extraction. The time to extract a watermark is the time to run step
2 and 3 in watermark extraction, {\em i.e.} candidate selection and watermark
identification.

To accelerate the extraction process, we parallelize the key steps across
multiple servers. More specifically, in the candidate selection step, any
available servers are assigned an unchecked watermark node to find its candidates. In
step 3, each available server will be assigned to search one watermark from
one candidate of watermark node $x_1$. When a watermark is found or no more
candidates are unchecked, the extraction process stops (for that user).

We perform measurements to quantify the actual impact of parallelizing
extraction over a cluster. All system parameters are the same as previous
tests, except that we embed $1$ watermark into a graph. \fixhan{To extract
watermarks, we compare the improved watermark extraction method to the
basic extraction method, with bucket size 10 and 
 NSD similarity of 0.75 in the improved extraction method.} In addition to the
two graphs, {\em i.e.} Flickr and LA, we also measure efficiency on the
largest graph in our study  (Livejournal, 5.2 million nodes, 49 million
edges), shown in Table~\ref{tab:graphs}.  We parallelize watermark
extraction across $10$ servers, each with 2.33GHz Xeon servers with 192GB
RAM. All experiments repeat on $10$
different watermarked graphs, and the time is the average of the $10$
computation time.

First result in Table~\ref{tab:wttime} is that watermarking system is
efficient in embedding and extracting watermarks. On average, embedding one
watermark into a graph is very fast. For example, average embedding
time for the largest graph, Livejournal, is around 12 minutes and embedding times
for Flickr and L.A. are less than 2 minutes. Even using one server to extract
watermarks, the computation time is small. Like in Flickr, the
extraction time is around 13 minutes \fixhan{using both the basic method and
  the improved method.}
From our observation, the time to identify the watermark graph on the
candidate subgraphs is much less than the time required to find and filter
candidates, which corresponds to around 99\% of total computation
time. Since finding candidates takes $O(kn)$ computational complexity
  and $k=(2+\delta)\log_{2}n$, the complexity to extract a watermark from a
  real-world graph is $O(n\log_2{n})$. Here $k$ is the number of nodes in the
  watermark graph and $n$ is the number of nodes in the total graph. 

Second, we find that speedup from distributed extraction is quite good,
with speedup of $8$ over 10 servers for Livejournal and $7$ for LA
  (for both extraction methods). The speedup for Flickr is only around $4$
  using both methods, because one of the watermarked graphs takes much longer
  time than others in finding candidates, $\sim 10$ minutes. This is almost 4
  times longer than the average extraction time on the other graphs. Not
  counting this outlier, average parallel extraction time on Flickr is around
  $150$
  seconds for both methods, which is $5$ times faster than using one
  server. This is because the core computation is finding candidates, and
  completion time can vary when computing the similarity of NSD between
  watermark nodes and graphs nodes, which depends on node degree. The higher
  the degree is, the longer it takes for the similarity computation. Since
  there are several Flickr watermark nodes of high degree, time to find
  candidates is relatively longer.

Finally, comparing the two extraction methods, there is no significant
difference between their computation time. This is because the extraction
time of both methods are dominated by the time to find and filter
candidates, which is $O(n\log_2{n})$ for both methods.

\fixhan{\para{Summary.} We evaluate the efficiency of the graph watermark
  embedding and extraction algorithms on three real-world graphs with
  $600K\sim5M$ nodes and $7M\sim48M$ edges. The results show that the
  embedding process is fast even for large graphs, and only takes up to 12
  minutes to embed a watermark into a graph with $5M$ nodes.  In the
  extraction process, the time to identify watermark graphs on the set of
  pre-filtered candidate nodes is much less than the time to filter candidate
  nodes, whose complexity is $O(n\log_2{n})$.  Our experimental results also
  show that on a single commodity server, the extraction time is at most
  $43$ minutes in a $5M$-node graph, and can be future reduced to less than
  $5$ minutes by distributing the computation across multiple servers.  }

\begin{table}[t]
\centering
\caption{The efficiency of the watermarking system, including watermark
  embedding time on one server, the extraction time on one server and the
  parallel extraction time across $10$ servers using basic watermark
  extraction method and improved watermark extraction method.}
\label{tab:wttime}
\resizebox{1.\columnwidth}{!}{
\begin{tabular}{|c| c | c c | c c|}
\hline
\multirow{2}{*} {Graph} & \multirow{2}{*}{Embedding (s)} &\multicolumn{2}{|c|}
{Basic Extraction} & \multicolumn{2}{|c|} {Improved Extraction}\\
\cline{3-6}
{} & {} & Single(s) &Parallel (s) & Single (s) & Parallel (s)\\
\hline
LA & 40 & 270 & 39 & 310 & 42 \\
\hline
Flickr & 80 & 767 & 195 & 776 & 197\\
\hline
Livejournal & 695 & 2568 & 310 &  2605 & 317\\
\hline
\end{tabular}}
\end{table}

\section{Conclusion}
In this paper, we take a first step towards the design and implementation of
a robust graph watermarking system.  Graph watermarks have the potential to
significantly impact the way graphs are shared and tracked.  Our work
identifies the critical requirements of such a system, and provides an initial
design that targets the critical properties of uniqueness, robustness to
attacks, and minimal distortion to the graph structure.  We also identify key
attacks against graph watermarks, and evaluate them against an improved
design with additional features for improved robustness under attack.

Our evaluation shows that our initial watermarking system modifies very few
nodes and edges in a graph, {\em i.e.} less than 0.04\% nodes and edges in a
graph with 603K nodes and 7.6M edges. Results also demonstrate extremely low
distortion, {\em i.e.} the watermarked graphs are highly consistent with the
original graph in all graph metrics we considered.  Empirical tests on
several real, large graphs show that our robustness features dramatically
improved our resilience against both single and multi-user collusion attacks.
Finally, we show that the embedding process and the extraction process are
efficient, and the extraction process is easily parallelized over a computing
cluster.

While our proposed scheme achieves many of our initial goals, there is
significant room for improvement and ongoing work.  One focus is developing
stronger redundancy schemes to protect against attackers with a greater
tolerance for graph distortion, {\em i.e.} willing to make a greater number
of node/edge changes.  Another is to develop alternate schemes that can
recover more information about multiple attackers in the colluding attack
model.



\balance
\begin{small}
\bibliographystyle{acm}
\bibliography{zhao,social,han,p2p}
\end{small}

\end{document}